\PassOptionsToPackage{dvips,final}{graphicx}

\documentclass{article}

\usepackage{graphicx}
\usepackage{latexsym}
\usepackage{amsmath,amssymb}
\usepackage{xspace}
\usepackage[english]{babel}
\usepackage[latin1]{inputenc}

\usepackage{mathrsfs}
\usepackage{times}

\usepackage[final]{hyperref}

\newtheorem{theorem}{Theorem}[section]
\newtheorem{corollary}[theorem]{Corollary}
\newtheorem{lemma}[theorem]{Lemma}
\newtheorem{proposition}[theorem]{Proposition}
\newtheorem{fact}[theorem]{Fact}

\newtheorem{example}[theorem]{Example}
\newtheorem{definition}[theorem]{Definition}
\newtheorem{remark}[theorem]{Remark}
\newenvironment{proof}{\vspace{0mm}\noindent {\bf Proof.}}{\vspace{2mm}}

\newcommand{\spath}[1]{\mbox{[$#1$]-}path}
\newcommand{\adj}[1]{\mbox{[$#1$]-}adjacent}
\newcommand{\component}[1]{\mbox{[$#1$]-}component}

\newcommand{\connected}[1]{\mbox{[$#1$]-}connected}
\newcommand{\touches}[1]{\mbox{[$#1$]-}touches}

\newcommand{\nodes}{\mathit{nodes}}
\newcommand{\edges}{\mathit{edges}}
\newcommand{\HD}{H\!D}
\newcommand{\HG}{{\cal H}}
\newcommand{\JT}{J\!T}

\renewcommand{\root}{\mathit{root}}
\newcommand{\vertices}{\mathit{vertices}}

\newcommand{\NP}{\mbox{\rm NP}}

\newcommand{\V}{\mathcal{V}}

\newcommand{\CR}{\mbox{\rm R\&C}}

\newcommand{\E}{\mathrm{ED}}
\newcommand{\F}{\mathrm{Fr}}
\newcommand{\ecomponent}[1]{\mbox{[$#1$]-}option}
\newcommand{\icomponent}[1]{\mbox{\em [$#1$]-}component}
\newcommand{\itouches}[1]{\mbox{\em [$#1$]-}touches}
\newcommand{\iecomponent}[1]{\mbox{\em [$#1$]-}option}

\newcommand{\treecomp}{\mathit{C_{\mbox{\tiny $\top$}}}\hspace{-0,5mm}}

\newcommand{\tuple}[1]{\langle#1\rangle}

\newcommand{\nop}[1]{}

\newcommand{\longv}[1]{}

\begin{document}

\title{{Tree Projections and Structural Decomposition Methods}:\\~Minimality~and~Game-Theoretic~Characterization}

\author{Gianluigi Greco and Francesco Scarcello\\
\ \\
\small  University of Calabria, 87036, Rende, Italy\\
\small  {\tt \{ggreco\}@mat.unical.it}, {\tt \{scarcello\}@deis.unical.it} }

\date{}

\maketitle

\begin{abstract}
Tree projections provide a mathematical framework that encompasses all the various (purely) structural decomposition methods that have been
proposed in the literature to single out classes of nearly-acyclic (hyper)graphs, such as the \emph{tree decomposition method}, which is the
most powerful decomposition method on graphs, and the \emph{(generalized) hypertree decomposition method}, which is its natural counterpart on
arbitrary hypergraphs.

The paper analyzes this framework, by focusing in particular on ``minimal'' tree projections, that is, on tree projections without useless
redundancies. First, it is shown that minimal tree projections enjoy a number of properties that are usually required for normal form
decompositions in various structural decomposition methods. In particular, they enjoy the same kind of connection properties as (minimal) tree
decompositions of graphs, with the result being tight in the light of the negative answer that is provided to the open question about whether
they enjoy a slightly stronger notion of connection property, defined to speed-up the computation of hypertree decompositions.
Second, it is shown that tree projections admit a natural game-theoretic characterization in terms of the Captain and Robber game. In
this game, as for the Robber and Cops game characterizing \emph{tree decompositions}, the existence of winning strategies implies the existence
of monotone ones. As a special case, the Captain and Robber game can be used to characterize the generalized hypertree decomposition method,
where such a game-theoretic characterization was missing and asked for.
Besides their theoretical interest, these results have immediate algorithmic applications both for the general setting and for structural
decomposition methods that can be recast in terms of tree projections.
\end{abstract}

\raggedbottom

\section{Introduction}

\subsection{Structural Decomposition Methods and Open Questions}

Many $\NP$-hard problems in different application areas, ranging, e.g., from AI~\cite{gott-etal-00} to Database Theory~\cite{bern-good-81}, are
known to be efficiently solvable when restricted to instances whose underlying structures can be modeled via acyclic graphs or hypergraphs.
Indeed, on these kinds of instances, solutions can usually be computed via dynamic programming, by incrementally processing the acyclic
(hyper)graph, according to some of its topological orderings.
However, structures arising from real applications are hardly precisely acyclic. Yet, they are often not very intricate and, in fact, tend to
exhibit some limited degree of cyclicity, which suffices to retain most of the nice properties of acyclic ones. Therefore, several efforts have
been spent to investigate invariants that are best suited to identify nearly-acyclic graph/hypergraphs, leading to the definition of a number
of so-called {\em structural decomposition methods}, such as the \emph{(generalized) hypertree}~\cite{gott-etal-99}, \emph{fractional
hypertree}~\cite{GM06}, \emph{spread-cut}~\cite{CJG08}, and \emph{component hypertree}~\cite{GMS07} decompositions. These methods aim at
transforming a given cyclic hypergraph into an acyclic one, by organizing its edges (or its nodes) into a polynomial number of clusters, and by
suitably arranging these clusters as a tree, called decomposition tree. The original problem instance can then be evaluated over such a tree of
subproblems, with a cost that is exponential in the cardinality of the largest cluster, also called {\em width} of the decomposition, and
polynomial if this width is bounded by some constant.

\begin{figure*}[t]
  \centering
  \includegraphics[width=0.99\textwidth]{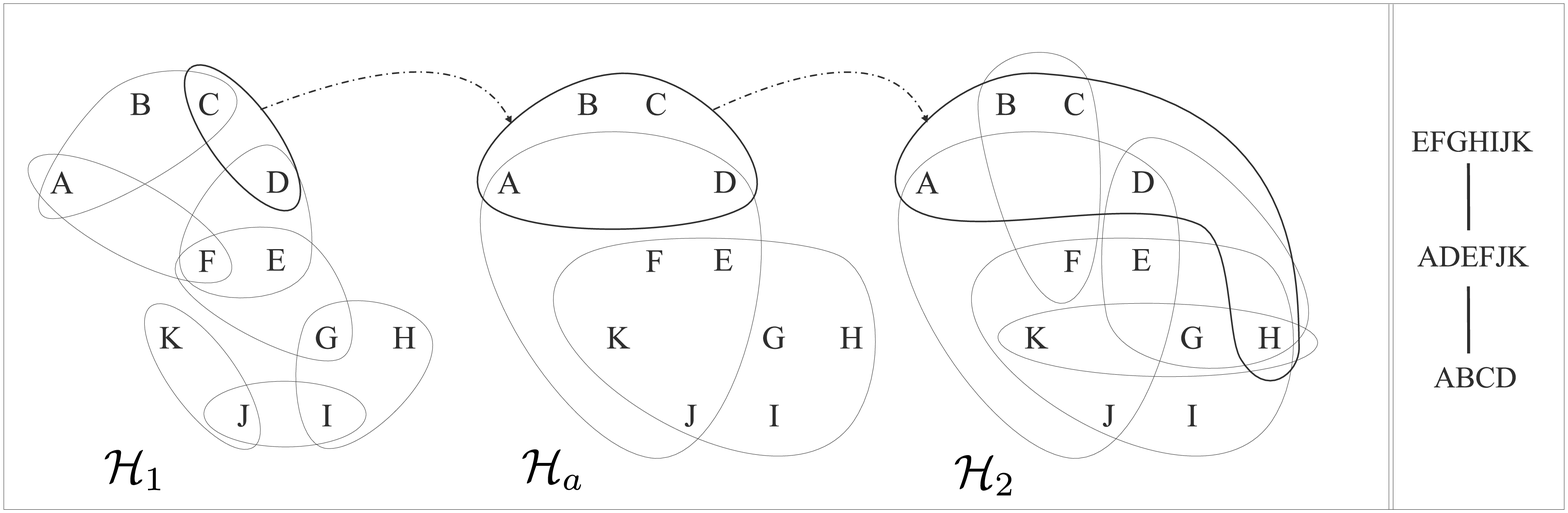}\vspace{-2mm}
  \caption{A Tree Projection $\HG_a$ of $\HG_1$ with respect $\HG_2$; e.g., $\{C,D\}\subseteq \{A,B,C,D\}\subseteq\{A,B,C,D,H\}$.
  On the right: A Join Tree $\JT_a$ for $\HG_a$.
  }\label{fig:hypegraph}\vspace{-3mm}
\end{figure*}

Despite their different technical definitions, there is a simple mathematical framework that encompasses all \emph{purely} structural
decomposition methods, which is the framework of the \emph{tree projections}~\cite{GS84}.
Roughly, given a pair of hypergraphs $(\HG_1,\HG_2)$, a tree projection of $\HG_1$ w.r.t. $\HG_2$ is an acyclic hypergraph $\HG_a$ such that
each hyperedge of $\HG_1$ is contained in some hyperedge of $\HG_a$, that is in its turn contained in a hyperedge of $\HG_2$, which is called
the {\em resource hypergraph}---see Figure~\ref{fig:hypegraph} for an illustration.

Therefore, in the tree projection framework, the resource hypergraph $\HG_2$ is arbitrary. Whenever it is instead computed with some specific
technique from the hypergraph $\HG_1$, we obtain as special cases the so-called \emph{purely} structural decomposition methods. Consider, for
instance, the \emph{tree decomposition} method~\cite{DP89,Fre90}, based on the notion of treewidth~\cite{RS84}, which is the most general
decomposition method over classes of graphs (see, e.g, \cite{gott-etal-00,G07}). Let $k$ be a fixed natural number, and consider any
(hyper)graph $\HG_1$ over a set $\V$ of nodes. Let $\HG_1^{tk}$ be the hypergraph associated whose hyperedges are all possible sets of at most
$k+1$ variables. Then, a hypergrah $\HG_1$ has \emph{treewidth} bounded by $k$ {if, and only if,} there is a tree projection of $\HG_1$
w.r.t.~$\HG_1^{tk}$ (see, e.g., \cite{GS08,GS10}).\footnote{For the sake of completeness, observe that the only known structural technique that
does not fit the general framework of tree projections is the one based on the submodular width~\cite{M10}, which is not purely structural.
Indeed, this method, which is specifically tailored to solve constraint satisfaction problem (or conjunctive query) instances, identify a
number of decompositions on the basis of both the given constraint hypergraph and the associated constraint relations.}

In fact, our current understating of structural decompositions for binary (graph) instances is fairly complete. The situation pertaining
decompositions methods for arbitrary (hypergraphs) instances is much more muddled instead. In particular, the following two questions have been
posed in the literature for the general tree projection framework as well as for structural decomposition methods specifically tailored to deal
with classes of queries without a fixed arity bound. Such questions were in particular open for the generalized hypertree decomposition method,
which on classes of unbounded-arity queries is a natural counterpart of the tree decomposition method.

\medskip

\noindent \textbf{(Q2) Is there a natural notion of normal-form for tree projections?} Whenever some tree projection of a pair $(\HG_1,\HG_2)$
exists, in general there are also many tree projections with useless redundancies. Having a suitable notion of minimality may allow us to
identify the most desirable tree projections. In fact, for several structural decomposition methods, \emph{normal forms} have been defined to
restrict the search space of decomposition trees, without loosing any useful decomposition.

Such a nice feature is however missing for the general case of tree projections. As a consequence, consider for instance the basic problem of
deciding whether a tree projection of a hypergraph $\HG_1$ with respect to a hypergraph $\HG_2$ exists or not. Because every subset of any
hyperedge of $\HG_2$ may belong to the tree projection $\HG_a$ we are looking for, this latter hypergraph might in principle consists of an
exponential number of hyperedges (w.r.t.~to the size of $\HG_1$ and $\HG_2$). Therefore, even proving that the existence problem is feasible in
$\NP$ is not easy, without a notion of minimality that allows us to get rid of redundant hyperedges.

Furthermore, in the case of tree decompositions, it is known that we can focus, w.l.o.g., on \emph{connected} ones~\cite{FN06}, that is,
basically, on tree decompositions such that, for each set of connected vertices, the sub-hypergraph induced by the nodes covered in such
vertices is connected in its turn. Again, connected decompositions provide us with a ``normal form'' for decomposition trees, which can be
exploited to restrict the search space of the possible decompositions and, thus, to speed-up their computation~\cite{FN06}. However, no
systematic study about connection properties of tree projections (and of decomposition methods other than tree decomposition) appeared in the
literature. Algorithms have been implemented limiting the search space to a kind of connected (generalized) hypertree
decompositions~\cite{SH07}, but it was left open whether the resulting method is a heuristic one or it does give an exact~solution.

\medskip

\noindent \textbf{(Q2) Is there a natural game-theoretic characterization for tree projections?} Tree decompositions have a nice
{game-theoretic characterization} in terms of the  \emph{Robber and Cops} game~\cite{ST93}: A hypergraph $\HG$ has treewidth bounded by $k$ if,
and only, if $k+1$ Cops can capture a Robber that can run at great speed along the hyperedges of $\HG$, while being not permitted to run trough
a node that is controlled by a Cop. In particular, the Cops can move over the nodes, and while they move, the Robber is fast and can run trough
those nodes that are left or not yet occupied before the move is completed. An important property of this game is that there is no restriction
on the strategy used by the Cops to capture the Robber. In particular, the Cops are not constrained to play {\em monotone strategies}, that is,
to shrink the Robber's escape space in a monotonically decreasing way. More precisely, playing non-monotone strategies gives no more power to
the Cops~\cite{ST93}. In many results about treewidth (e.g., \cite{ABD07}), this property turns out to be very useful, because good strategies
for the Robber may be easily characterized as those strategies that allow the Robber to run forever.

\emph{Hypertree decomposition} is an efficiently recognizable structural method~\cite{gott-etal-03}, which provides a 3-approximation for
generalized hypertree decompositions~\cite{AGG07}. This method is also known to have a nice game-theoretic characterization, in terms of the
(monotone) \emph{Robber and Marshals} game~\cite{gott-etal-03}, which can be viewed as a natural generalization of the Robber and Cops games.
The game is the same as the one characterizing acyclicity, but with $k$ Marshals acting simultaneously to capture the Robber: A hypergraph
$\HG$ has \emph{hypertree width} bounded by $k$ if, and only if, $k$ Marshals, each one with the ability of controlling a hyperedge of $\HG$,
can capture a Robber that can run at great speed along the hyperedges, while being not permitted to run trough a node that belongs to a
hyperedge controlled by a Marshal. Note that Marshals are more powerful than the Cops of the Robber and Cops game characterizing treewidth, in
that they can move over whole hyperedges. However, Marshals are now required to play monotonically, because non-monotone strategies give some
extra-power that does not correspond to valid decompositions~\cite{adler04}.

Despite the similarities between hypertree and generalized hypertree decompositions as they are apparent from the original definitions by
Gottlob et al.~\cite{gott-etal-03}, game theoretic characterizations for generalized hypertree width were still missing. In~\cite{adler04}, it
is raised the question about whether there is a (natural) game theoretic characterization for generalized hypertree width, where
non-monotonicity does not represent a source of additional power. Such a characterization is missing for the tree projection setting, too.

\subsection{Contributions}

In this paper, we provide useful properties and characterizations of tree projections (and structural decomposition methods), by answering the
two questions illustrated above. In particular,

\begin{itemize}
    \item[$\blacktriangleright$] We define and investigate \emph{minimal} tree projections, where the minimal possible subsets of any view
        are employed. Intuitively, such tree projections typically correspond to more efficient decompositions. We show that some
        properties required for ``normal form'' decompositions in various notions of structural decomposition methods (see,
        e.g.,~\cite{gott-etal-99}) are a consequence of minimality. In particular, minimal tree projections enjoy the same kind of
        connection property as tree decompositions.

    \item[$\blacktriangleright$] We define a normal form for (minimal) tree projections. In particular, it turns out that, given any pair
        of hypergraphs $(\HG_1,\HG_2)$,  there always exists a tree projection of $\HG_1$ w.r.t.~$\HG_2$ in normal form having polynomial
        size with respect to the size of the given hypergraphs. An immediate consequence of this result is that checking whether a tree
        projection exists or not is feasible in $\NP$. In fact, this property has already been exploited in the $\NP$-completeness proof of
        tree projections~\cite{GMS07}.

   \item[$\blacktriangleright$] We give a negative answer to the question raised in~\cite{SH07} for (generalized) hypertree decomposition.
       We observe that the notion of connected decomposition proposed there differs from the one defined for tree decompositions and
       mentioned above. In particular, we show a hypergraph where this restriction leads to worse tree projections, more precisely, where
       all such connected (generalized) hypertree decompositions have width higher than the hypertree width of the considered hypergraph.
       Hence, the algorithm proposed in~\cite{SH07} for connected hypertree decompositions is not complete, as far as the computation of
       unrestricted decompositions is considered.

  \item[$\blacktriangleright$] We define the \emph{Captain and Robber} game to be played on pairs of hypergraphs, and we show that in this
      game the Captain has a winning strategy if, and only if, she has a monotone winning one. Then, we show that tree projections and
      thus, e.g., generalized hypertree decompositions, may be characterized in terms of the Captain and Robber game. Hence, these notions
      have now a natural game characterization where monotone and non-monotone strategies have the same power.
\end{itemize}

\medskip

\noindent \textbf{Organization.} The rest of the paper is organized as follows.
Section~\ref{sec:framework} illustrates some basic notions and concepts. The setting of minimal tree projections is discussed in
Section~\ref{MTP}. The game-theoretic characterization is illustrated in Section~\ref{GTC}. A few final remarks and some further results are
discussed in Section~\ref{sec:applications}, by exploiting the properties of minimal tree projections and the game-theoretic characterization.

\section{Preliminaries}\label{sec:framework}

\noindent \textbf{Hypergraphs and Acyclicity.} A \emph{hypergraph} $\HG$ is a pair $(V,H)$, where $V$ is a finite set of nodes and $H$ is a set
of hyperedges such that, for each $h\in H$, $h\subseteq V$.
If $|h|=2$ for each (hyper)edge $h\in H$, then $\HG$ is a {\em graph}.
For the sake of simplicity, we always denote $V$ and $H$ by $\nodes(\HG)$ and $\edges(\HG)$, respectively.

A hypergraph $\HG$ is {\em acyclic} (more precisely, $\alpha$-acyclic~\cite{fagi-83}) if, and only if, it has a join tree~\cite{bern-good-81}.
A {\em join tree} $\JT$ for a hypergraph $\HG$ is a tree whose vertices are the hyperedges of $\HG$ such that, whenever a node $X\in V$ occurs
in two hyperedges $h_1$ and $h_2$ of $\HG$, then $h_1$ and $h_2$ are connected in $\JT$, and $X$ occurs in each vertex on the unique path
linking $h_1$ and $h_2$ (see Figure~\ref{fig:hypegraph} for an illustration). In words, the set of vertices in which $X$ occurs induces a
(connected) subtree of $\JT$. We will refer to this condition as the {\em connectedness condition} of join trees.

\medskip \noindent \textbf{Tree Decompositions.}
A \emph{tree decomposition}~\cite{RS84} of a graph $G$ is a pair $\tuple{T,\chi}$, where $T=(N,E)$ is a tree, and $\chi$ is a labeling function
assigning to each vertex $v\in N$ a set of vertices $\chi(v)\subseteq \nodes(G)$, such that the following conditions are satisfied: (1) for
each node $Y\in \nodes(G)$, there exists $p\in N$ such that $Y\in \chi(p)$; (2) for each edge  $\{X,Y\}\in \edges(G)$, there exists $p\in N$
such that $\{X,Y\}\subseteq \chi(p)$; and (3) for each node $Y\in \nodes(G)$, the set $\{p\in N \mid  Y\in \chi(p)\}$ induces a (connected)
subtree of $T$. The \emph{width} of $\tuple{T,\chi}$ is the number $\max_{p\in N}(|\chi(p)|-1)$.

The \emph{Gaifman graph} of a hypergraph $\HG$ is defined over the set $\nodes(\HG)$ of the nodes of $\HG$, and contains an edge $\{X,Y\}$ if,
and only if, $\{X,Y\}\subseteq h$ holds, for some hyperedge $h\in\edges(\HG)$. The {\em treewidth} of $\HG$ is the minimum width over all the
tree decompositions of its Gaifman graph. Deciding whether a given hypergraph has treewidth bounded by a fixed natural number $k$ is known to
be feasible in linear time~\cite{bodl-96}.

\medskip \noindent \textbf{(Generalized) Hypertree Decompositions.}
A {\em hypertree for a hypergraph $\HG$} is a triple $\tuple{T,\chi,\lambda}$, where $T=(N,E)$ is a rooted tree, and $\chi$ and $\lambda$ are
labeling functions which associate each vertex $p\in N$ with two sets $\chi(p)\subseteq \nodes(\HG)$ and $\lambda(p)\subseteq \edges(\HG)$. If
$T'=(N',E')$ is a subtree of $T$, we define $\chi(T')= \bigcup_{v\in N'} \chi(v)$.
In the following, for any rooted tree $T$, we denote the set of vertices $N$ of $T$ by $\vertices(T)$, and the root of $T$ by $\root(T)$.
Moreover, for any $p\in N$, $T_p$ denotes the subtree of $T$ rooted at $p$.

A {\em generalized hypertree decomposition}~\cite{gott-etal-03} of a hypergraph $\HG$ is a hypertree $\HD=\tuple{T,\chi,\lambda}$ for $\HG$
such that: (1) for each hyperedge $h\in \edges(\HG)$, there exists $p\in \vertices(T)$ such that $h\subseteq \chi(p)$; (2) for each node $Y\in
\nodes(\HG)$, the set $\{ p\in \vertices(T) \mid  Y \in \chi(p) \}$ induces a (connected) subtree of $T$; and (3) for each $p\in \vertices(T)$,
$\chi(p)\subseteq \nodes(\lambda(p))$.
The {\em width} of a generalized hypertree decomposition $\tuple{T,\chi,\lambda}$ is $max_{p\in \vertices(T)} |\lambda(p)|$. The {\em
generalized hypertree width} $ghw(\HG)$ of $\HG$ is the minimum width over all its
generalized hypertree decompositions. 

A \emph{hypertree decomposition}~\cite{gott-etal-99} of $\HG$ is a generalized hypertree decomposition $\HD=\tuple{T,\chi,\lambda}$ where: (4)
for each $p\in \vertices(T)$, $\nodes(\lambda(p)) \cap \chi(T_p) \;\subseteq\; \chi(p)$. Note that the inclusion in the above condition is
actually an equality, because Condition~(3) implies the reverse inclusion. The {\em hypertree width} $hw(\HG)$ of $\HG$ is the minimum width
over all its hypertree decompositions.
Note that, for any hypergraph $\HG$, it is the case that $ghw(\HG)\leq hw(\HG)\leq 3\times ghw(\HG)+1$~\cite{AGG07}. Moreover, for any fixed
natural number $k>0$, deciding whether $hw(\HG)\leq k$ is feasible in polynomial time (and, actually, it is
highly-parallelizable)~\cite{gott-etal-99}, while deciding whether $ghw(\HG)\leq k$ is $\NP$-complete~\cite{GMS07}.

\medskip \noindent\textbf{Tree Projections.}
For two hypergraphs $\HG_1$ and $\HG_2$, we write $\HG_1\leq \HG_2$ if, and only if, each hyperedge of $\HG_1$ is contained in at least one
hyperedge of $\HG_2$. Let $\HG_1\leq \HG_2$; then, a \emph{tree projection} of $\HG_1$ with respect to $\HG_2$ is an acyclic hypergraph $\HG_a$
such that $\HG_1\leq \HG_a \leq \HG_2$. Whenever such a hypergraph $\HG_a$ exists, we say that the pair of hypergraphs $(\HG_1,\HG_2)$ has a
tree projection.

Note that the notion of tree projection is more general than the above mentioned (hyper)graph based notions. For instance, consider the
generalized hypertree decomposition approach. Given a hypergraph $\HG$ and a natural number $k>0$, let $\HG^k$ denote the hypergraph over the
same set of nodes as $\HG$, and whose set of hyperedges is given by all possible unions of $k$ edges in $\HG$, i.e., $\edges(\HG^k)= \{ h_1
\cup h_2 \cup \cdots \cup h_k \mid \{h_1,h_2,\ldots,h_k\}\subseteq \edges(\HG)\}$. Then, it is well known and easy to see that $\HG$ has
generalized hypertree width at most $k$ if, and only if, there is a tree projection for $(\HG,\HG^k)$.

Similarly, for tree decompositions, let $\HG^{tk}$ be the hypergraph over the same set of nodes as $\HG$, and whose set of hyperedges is given
by all possible clusters $B\subseteq\nodes(\HG)$ of nodes such that $|B| \leq k+1$. Then, $\HG$ has treewidth at most $k$ if, and only if,
there is a tree projection for $(\HG,\HG^{tk})$.

\section{Minimal Tree Projections}\label{MTP}

In this section, a partial ordering of tree projections is defined. It is shown that minimal tree projections have nice properties with both
theoretical and practical interest.

Let $\HG$ and $\HG'$ be two hypergraphs. We say that $\HG$ is \emph{contained} in $\HG'$, denoted by $\HG\subseteq \HG'$, if for each hyperedge
$h\in \edges(\HG)-\edges(\HG')$, there is a hyperedge $h'\in \edges(\HG')-\edges(\HG)$ with $h\subseteq h'$ (and hence $h\subset h'$).
Moreover, we say that $\HG$ is \emph{properly contained} in $\HG'$, denoted by $\HG\subset \HG'$, if $\HG\subseteq \HG'$ and $\HG\neq \HG'$.

Note that $\edges(\HG)\subseteq \edges(\HG')$ entails $\HG\subseteq \HG'$ (and hence $\edges(\HG)\subset \edges(\HG')$ entails $\HG\subset
\HG'$). Moreover, $\HG\subseteq \HG'$ implies $\HG\leq \HG'$, but the converse is not true. For example, if $\edges(\HG)=\{h_1,h_2\}$ with
$h_2\subset h_1$ and $\edges(\HG')=\{h_1\}$, then $\HG'\subseteq \HG$ and $\HG'\leq \HG$ hold, as $\edges(\HG')\subset \edges(\HG)$. Moreover,
$\HG\leq \HG'$ holds too, but $\HG$ is not contained in $\HG'$ as there is no hyperedge $h'\in\edges(\HG')-\edges(\HG)$ such that $h_2\subseteq
h'$.

\begin{definition}\label{def:minimal}\em
Let $\HG_1$ and $\HG_2$ be two hypergraphs.
Then, a tree projection $\HG_a$ for $(\HG_1,\HG_2)$ is \emph{minimal} if there is no tree projection $\HG_a'$ of $\HG_1$ wr.t. $\HG_2$ with
$\HG_a'\subset \HG_a$.~\hfill~$\Box$
\end{definition}

\subsection{Basic Facts}

We first point out a number of basic important properties of tree projections of a given pair of hypergraphs $(\HG_1,\HG_2)$.

\begin{fact}\label{fact:ordering}
The relationship $\subseteq$ of Definition~\ref{def:minimal} induces a partial ordering over the tree projections of $\HG_1$ w.r.t.~$\HG_2$.
\end{fact}

\begin{proof}
Observe first that the relation `$\subseteq$' over hypergraphs is \emph{reflexive}. We next show that it is \emph{antisymmetric}, too. Let
$\HG_1$ and $\HG_2$ be two hypergraphs such that $\HG_1\subseteq \HG_2$ and $\HG_2\subseteq \HG_1$, and assume by contradiction that $\HG_1\neq
\HG_2$. Thus, $\edges(\HG_2)\neq \edges(\HG_1)$. Moreover, $\edges(\HG_2)\not\supset \edges(\HG_1)$ holds, for otherwise it is trivially
impossible that $\HG_2\subseteq \HG_1$. Then, let $h_1$ be the largest hyperedge (with the maximum number of nodes) in $\edges(\HG_1)\setminus
\edges(\HG_2)$. Since $\HG_1\subseteq \HG_2$, it is the case that there is a hyperedge $h_2\in \edges(\HG_{2})\setminus \edges(\HG_{1})$ with
$h_1\subset h_2$. But we also know that $\HG_2\subseteq \HG_1$ holds, and hence there is a hyperedge $h_1'\in \edges(\HG_1)\setminus
\edges(\HG_2)$ with $h_2\subset h_1'$. Thus, $h_1\subset h_2\subset h_1'$, which is impossible due to the maximality of $h_1$.

Eventually, we show that the relation `$\subseteq$' over hyperedges is \emph{transitive}. Indeed, assume $\HG_1\subseteq \HG_2$ and
$\HG_2\subseteq \HG_3$. Let $h_1$ be a hyperedge in $\edges(\HG_{1})\setminus \edges(\HG_{3})$. We distinguish two cases. If $h_1\in
\edges(\HG_2)$, and hence $h_1\in \edges(\HG_2)\setminus \edges(\HG_3)$, then there is a hyperedge $h_3\in \edges(\HG_3)\setminus
\edges(\HG_2)$ such that $h_1\subseteq h_3$. Otherwise, i.e., if $h_1\not\in \edges(\HG_2)$, and hence $h_1\in \edges(\HG_1)\setminus
\edges(\HG_2)$, then there is a hyperedge $h_2\in \edges(\HG_2)\setminus \edges(\HG_1)$ such that $h_1\subseteq h_2$. Then, we have to consider
two subcases. If $h_2\in \edges(\HG_3)$, then we have that $h_2$ is actually a hyperedge in $\edges(\HG_3)\setminus \edges(\HG_1)$ such that
$h_1\subseteq h_2$.
Instead, if $h_2\not\in \edges(\HG_3)$, and hence $h_2\in\edges(\HG_2)\setminus \edges(\HG_3)$, then there is a hyperedge $h_3'\in
\edges(\HG_3)\setminus \edges(\HG_2)$ with $h_2\subseteq h_3$. It follows that $h_1\subseteq h_2\subseteq h_3'$. Putting it all together, we
have shown that in all the possible cases, for each hyperedge $h_1\in \edges(\HG_{1})\setminus \edges(\HG_{3})$ there is a hyperedge $h'\in
\edges(\HG_3)\setminus \edges(\HG_1)$ such that $h_1\subseteq h'$. It follows that $\HG_1\subseteq \HG_3$ holds.

By the above properties, `$\subseteq$' is a partial order, and `$\subset$' is a strict partial order over hypergraphs.\hfill $\Box$
\end{proof}

Hence, minimal tree projections always exist, as long as a tree projection exists.

\begin{fact}\label{thm:reducedMinimal}
The pair $(\HG_1,\HG_2)$ has a tree projection if, and only if, it has a minimal tree projection.
\end{fact}

A further property (again rather intuitive) is that minimal tree projections are \emph{reduced} hypergraphs. Recall that a hypergraph $\HG_a$
is \emph{reduced} if $\edges(\HG_a)$ does not contain two hyperedges $h_a$ and $\bar h_a$ such that $h_a\subset \bar h_a$.

\begin{fact}\label{thm:reducedMinimal2}
Every minimal tree projection is reduced.
\end{fact}

\begin{proof}
Assume for the sake of contradiction that $\HG_a$ is a minimal tree projection of $\HG_1$ w.r.t.~$\HG_2$ such that $\HG_a$ is not reduced. Let
$h_a$ and $\bar h_a$ be two hyperedges of $\HG_a$ such that $h_a\subset \bar h_a$. Consider the tree projection $\HG_a'\neq \HG_a$ obtained by
removing $h_a$ from $\HG_a$, and notice that $\HG_a'\leq \HG_a\leq \HG_2$ and $\HG_1\leq \HG_a'$. Thus, $\HG_a'$ is a tree projection of
$\HG_1$ w.r.t.~$\HG_2$.
However, we have that $\edges(\HG_a')\subset \edges(\HG_a)$, which entails that $\HG_a'\subset \HG_a$ holds, thereby contradicting the
minimality of $\HG_a$.\hfill $\Box$
\end{proof}

The last basic fact is rather trivial: minimal tree projections do not contain nodes that do not occur in $\HG_1$.

\begin{fact}\label{fact:sameNodes}
Let $\HG_a$ be a minimal tree projection of $\HG_1$ w.r.t.~$\HG_2$. Then, $\nodes(\HG_a)=\nodes(\HG_1)$.
\end{fact}

\begin{proof}
Let $\HG_a$ be a minimal tree projection of $\HG_1$ w.r.t.~$\HG_2$. Of course, $\nodes(\HG_a)\supseteq\nodes(\HG_1)$ clearly holds. On the
other hand, if $\nodes(\HG_a)\supset\nodes(\HG_1)$, the hypergraphs $\HG_a'$ obtained by deleting from every hyperedge each node in
$\nodes(\HG_a)\setminus\nodes(\HG_1)$ is still an acyclic hypergraph, and $\HG_1\leq\HG_a'\leq\HG_a$ holds. Moreover, it is straightforward to
check that $\HG_a'\subset\HG_a$, which contradicts the minimality of $\HG_a$. \hfill $\Box$
\end{proof}

\subsection{Component trees}

We now generalize to the setting of tree projections some properties of join trees that are required for efficiently computable decompositions
in various notions of structural decomposition methods (see, e.g.,~\cite{gott-etal-99}).
To formalize these properties, we need to introduce some additional definitions, which will be intensively used in the following.

Assume that a hypergraph $\HG$ is given. Let $V$, $W$, and $\{X,Y\}$ be sets of nodes. Then, $X$ is said \adj{V}\ (in $\HG$) to $Y$ if there
exists a hyperedge $h\in \edges(\HG)$ such that $\{X,Y\}\subseteq (h -V)$. A \spath{V}\ from $X$ to $Y$ is a sequence $X=X_0,\ldots,X_\ell=Y$
of nodes such that $X_{i}$ is \adj{V}\ to $X_{i+1}$, for each $i\in [0...\ell\mbox{-}1]$. We say that $X$ \touches{V} $Y$ if $X$ is
\adj{\emptyset} to $Z\in \nodes(\HG)$, and there is a \spath{V} from $Z$ to $Y$; similarly, $X$ \touches{V} the set $W$ if $X$ \touches{V} some
node $Y\in W$.
We say that $W$ is \connected{V}\ if $\forall X,Y\in W$ there is a \spath{V}\ from $X$ to $Y$. A \component{V} (of $\HG$) is a maximal
\connected{V}\ non-empty set of nodes $W\subseteq (\nodes(\HG)-V)$. For any \component{V}\ $C$, let $\edges(C) = \{ h\in \edges(\HG)\;|\; h\cap
C\neq\emptyset\}$, and for a set of hyperedges $H\subseteq \edges(\HG)$, let $\nodes(H)$ denote the set of nodes occurring in $H$, that is
$\nodes(H)=\bigcup_{h\in H} h$. For any component $C$ of $\HG$, we denote by $\F(C,\HG)$  the \emph{frontier} of $C$ (in $\HG$), i.e., the set
$\nodes(\edges(C))$.\footnote{The choice of the term ``frontier'' to name the union of a component with its outer border is due to the role
that this notion plays in hypergraph games, such as the one described in the subsequent section.}
Moreover, $\partial (C,\HG)$ denote the {\em border} of $C$ (in $\HG$), i.e., the set $\F(C,\HG)\setminus C$. Note that $C_1\subseteq C_2$
entails $\F(C_1,\HG)\subseteq \F(C_2,\HG)$. We write simply $\F(C)$ or $\partial C$, whenever $\HG$ is clear from the context.

We find often convenient to think at join trees as rooted trees: For each hyperedge $h\in \edges(\HG)$, the tree obtained by rooting $\JT$ at
vertex $h$ is denoted by $\JT[h]$ (if it is necessary to point out its root). Moreover, for each hyperedge $h'\in \edges(\HG)$ with $h'\neq h$,
let $\JT[h]_{h'}$ denote the subtree of $\JT[h]$ rooted at $h'$, and let $\nodes(\JT[h]_{h'})$ be the set of all nodes of $\HG$ occurring in
the vertices of $\JT[h]_{h'}$.

\begin{definition}\label{def:component-tree}\em
Let $\HG_1$ and $\HG_a$ be two hypergraphs with the same set of nodes such that $\HG_1\leq\HG_a$ and $\HG_a$ is acyclic. A join tree $\JT$ of
$\HG_a$, rooted at some vertex ${\it root}\in \edges(\HG_a)$, is said an {\it $\HG_1$-component tree} if the following conditions hold for each
vertex $h_r\in \edges(\HG_a)$ in $\JT$: \vspace{-1mm}
\begin{description}
  \item[{\sc subtrees$\mapsto$components.}] For each child $h_s$ of $h_r$ in $\JT$, there is exactly one \component{h_r} of $\HG_1$,
      denoted by $\treecomp(h_s)$, such that $\nodes(\JT_{h_s})=\treecomp(h_s)\cup (h_s\cap h_r)$. Moreover, $h_s\cap \treecomp(h_s)\neq
      \emptyset$ and $h_s\subseteq \F(\treecomp(h_s),\HG_1)$ hold.

  \item[{\sc components$\mapsto$subtrees.}] For each \component{h_r} $C_r$ of $\HG_1$ such that $C_r\subseteq \treecomp(h_r)$, with
      $\treecomp({\it root})$ being conventionally defined as $\nodes(\HG_1)$, there is exactly one child $h_s$ of $h_r$ in $\JT$ such that
      $C_r=\treecomp(h_s)$.\hfill $\Box$
\end{description}
\end{definition}

Interestingly, any reduced acyclic hypergraph $\HG_a$ has such an $\HG_a$-component tree (i.e., $\HG_1=\HG_a$, here), as pointed out in the
result below.

\begin{theorem}\label{thm:JTstrong}
Let $\HG_a$ be a reduced acyclic hypergraph (e.g., any minimal tree projection). For any hyperedge $h\in\edges(\HG_a)$, there exists a join
tree $\JT$ rooted at $h$ that is an $\HG_a$-component tree.
\end{theorem}

\begin{proof}
Let $\HG_a$ be any reduced acyclic hypergraph and let $h\in\edges(\HG_a)$ be any of its hyperedges, and consider
Definition~\ref{def:component-tree}, with its two parts: {\sc subtrees$\mapsto$components} and {\sc components$\mapsto$subtrees}.

\medskip
\noindent {\sc subtrees$\mapsto$components.}
We first show that there is a join tree $\JT$ for $\HG_a$ such that, for each pair $h_r,h_s\in \edges(\HG_a)$ where $h_s$ is a child of $h_r$
in $\JT[h]$, \vspace{-2mm}
\begin{enumerate}
  \item[\em (1)] there is exactly one \component{h_r} $C_r$ of $\HG_a$, denoted by $\treecomp(h_s)$, such that
      $\nodes(\JT[h]_{h_s})=\treecomp(h_s)\cup (h_s\cap h_r)$;

  \item[\em (2)] $h_s\cap \treecomp(h_s)\neq \emptyset$;

  \item[\em (3)] $h_s\subseteq \F(\treecomp(h_s),\HG_a)$.
\end{enumerate}

Since $\HG_a$ is a reduced acyclic hypergraph, the hypertree width of $\HG_a$ is 1. In particular, from the results in~\cite{gott-etal-99} (in
particular, from Theorem~5.4 in~\cite{gott-etal-99}) it follows that, for each hyperedge $h\in \edges(\HG_a)$, there is a width-1 hypertree
decomposition $\HD=\tuple{T,\chi,\lambda}$ for $\HG_a$, where $T$ is rooted at a vertex $\root(T)$ such that $\lambda(\root(T))=\{h\}$ and, for
each vertex $r\in \vertices(T)$ and for each child $s$ of $r$, the following conditions hold: (1) there is (exactly) one $\component{\chi(r)}$
$C_r$ of $\HG_a$ such that $\chi(T_s)\;=\; C_r\cup (\chi(s)\cap \chi(r))$; (2) $\chi(s)\cap C_r \neq \emptyset$, where $C_r$ is the
\component{\chi(r)} of $\HG_a$ satisfying Condition~(1); and (3) $h_s \cap \F(C_r,\HG_a) \neq\emptyset$ holds, where $\{h_s\}= \lambda(s)$ and
$C_r$ is the \component{\chi(r)} of $\HG_a$ satisfying Condition~(1).

Let us now denote by $h_p$ the unique (as the width is 1) hyperedge contained in $\lambda(p)$, for each vertex $p$ of $T$. Recall that
$h=h_{\root(T)}$ is the hyperedge associated with the root of $T$. Let $\JT[h]$ be the tree rooted at $h$ obtained from $T$ by replacing each
vertex $p$ with the corresponding hyperedge $h_p$. Then, for each vertex $r\in \vertices(T)$ and for each child $s$ of $r$, the three
conditions above that hold on $\HD$ can be rewritten as follows: (1) there is (exactly) one $\component{h_r}$ $C_r$ of $\HG_a$ such that
$\nodes(\JT[h]_{h_s})\;=\; C_r\cup (h_s\cap h_r)$; (2) $h_s\cap C_r \neq \emptyset$, where $C_r$ is the \component{h_r} of $\HG_a$ satisfying
Condition~(1); and (3) $h_s\subseteq \F(C_r,\HG_a)$, where $C_r$ is the \component{h_r} of $\HG_a$ satisfying Condition~(1).

It remains to show that $\JT$ is actually a join tree for $\HG_a$. To this end, we claim that the following two properties hold on $\HD$.\\

\noindent \textit{Property $P_1$:} \emph{$\forall p\in \vertices(T)$, $\chi(p)=\nodes(\lambda(p))$.}

\vspace{-2mm}
\begin{itemize}
\item[\ ] \textbf{Proof.} Recall that for each vertex $r\in \vertices(T)$ and for each child $s$ of $r$, the following conditions hold on
    the hypertree decomposition $\HD=\tuple{T,\chi,\lambda}$ for $\HG_a$: (1) there is (exactly) one $\component{\chi(r)}$ $C_r$ of $\HG_a$
    such that $\chi(T_s)\;=\; C_r\cup (\chi(s)\cap \chi(r))$; (2) $\chi(s)\cap C_r \neq \emptyset$, where $C_r$ is the \component{\chi(r)}
    of $\HG_a$ satisfying Condition~(1); and (3) $h_s \cap \F(C_r) \neq\emptyset$ holds, where $\{h_s\}= \lambda(s)$ and $C_r$ is the
    \component{\chi(r)} of $\HG_a$ satisfying Condition~(1).
In fact, $\chi(s)\not\subseteq \chi(r)$ holds, as $\chi(r)\cap C_r=\emptyset$ while $\chi(s)\cap C_r \neq \emptyset$.
Now, from Condition~(4) in the definition of hypertree decompositions it follows that, for each vertex $p\in \vertices(T)$,
$\chi(p)=\nodes(\lambda(p))\cap \chi(T_p)$. Thus, for each node $Y\in\nodes(\HG_a)$, the vertex $\bar p$ with $Y\in \chi(\bar p)$ that is
the closest to the root of $T$ is such that $\chi(\bar p)=\nodes(\lambda(\bar p))$. Indeed, each node $X\in h_{\bar p}$, where
$\lambda(\bar p)=\{h_{\bar p}\}$, must occur in the $\chi$-labeling of some vertex in the subtree rooted at $\bar p$ together with $Y$ in
order to satisfy Condition~(1) in the definition of hypertree decomposition. Thus, $X\in \chi(T_{\bar p})$.
Hence, for the vertex $\root(T)$, it is trivially the case that $\chi(\root(T))=\nodes(\lambda(\root(T)))$. Consider now an arbitrary
vertex $r\in \vertices(T)$ and let $s$ be a child of $r$. Thus, $\{h_s\}=\lambda(s)$, for some hyperedge $h_s$. Recall that
$\chi(s)\not\subseteq \chi(r)$, and take any node $Y\in h_s$ such that $Y\in \chi(s)\setminus \chi(r)$. Because of Condition~(2) in the
definition of hypertree decomposition, $Y$ cannot occur in the $\chi$-labeling of any vertex in path connecting $\root(T)$ and $r$ in $T$.
Thus, $s$ is the vertex closest to the root where $Y$ occurs. Hence, $\chi(s)=\nodes(\lambda(s))$. \hfill $\diamond$
\end{itemize}

\noindent \textit{Property $P_2$:} \emph{ $\forall p_1,p_2\in \vertices(T)$, $\lambda(p_1)\neq\lambda(p_2)$.}

\vspace{-2mm}
\begin{itemize}
\item[\ ] \textbf{Proof.} Assume for the sake of contradiction that there are two vertices $p_1$ and $p_2$ such that
    $\lambda(p_1)=\lambda(p_2)$. Because of Property $P_1$, $\nodes(\lambda(p_1))=\nodes(\lambda(p_2))=\chi(p_1)=\chi(p_2)$. Then, by
    Condition~(2) in the definition of hypertree decomposition, each vertex $p$ in the path between $p_1$ and $p_2$ is such that
    $\nodes(\lambda(p))=\chi(p)=\chi(p_1)=\chi(p_2)$ (because the hypergraph is reduced). In particular, this property holds for one vertex
    $r\in \vertices(T)$ and for one child $s$ of $r$. However, $\chi(r)=\chi(s)$ is impossible as we have observed in the proof of Property
    $P_1$. \hfill $\diamond$
\end{itemize}

Now, we show that hyperedges of $\HG_a$ one-to-one correspond to vertices of $\JT$, and that the connectedness condition holds on $\JT$.

For the first property, note that each vertex $p$ of $T$ corresponds to the hyperedge $h_p$, by construction. Moreover, by Property $P_2$, each
vertex of $\JT$ is mapped to a distinct hyperedge. Thus, it remains to show that for each hyperedge $\bar h\in \edges(\HG_a)$, there is a
vertex $p$ of $T$ such that $\bar h=h_p$. Indeed, note that by Condition~(1) of hypertree decompositions, for each hyperedge $\bar h\in
\edges(\HG_a)$, there is a vertex $p$ in $T$ such that $\bar h\subseteq \chi(p)$. By Property $P_1$ above, this entails that there is a
hyperedge $h_p\in \edges(\HG_a)$ such that $h_p=\chi(p)$ and $\bar h\subseteq h_p$. However, since $\HG_a$ is reduced, $\bar h=h_p$ holds.

We eventually observe that the connectedness condition holds on $\JT$. Indeed, if a node $Y\in\nodes(\HG_a)$ occurs in a vertex $h_{p}$ of
$\JT$, i.e., $Y\in h_p$, we have that $Y\in \chi(p)$ holds by Property $P_1$. By Condition~(2) of hypertree decompositions, the set $\{ p\in
\vertices(T) \mid Y \in \chi(p) \}$ induces a (connected) subtree of $T$. It follows that the set $\{ h_p\in \edges(\HG_a) \mid Y \in h_p \}$
induces a connected subtree of $\JT$.

\medskip

\noindent {\sc components$\mapsto$subtrees.}
Let us now complete the proof by showing that the join tree $\JT$ also satisfies the part {\sc components$\mapsto$subtrees} in
Definition~\ref{def:component-tree}.
Recall that $\treecomp(h)$ is defined as $\nodes(\HG_a)$ for the root $h$, and that $\treecomp(h_s)$ is the unique \component{h_r} with
$\nodes(\JT[h]_{h_s})=\treecomp(h_s)\cup (h_s\cap h_r)$, where $h_s$ is a child of $h_r$ in $\JT[h]$.
In fact, to conclude the proof, we next show that, for each vertex $h_r$ in $\JT[h]$ and for each \component{h_r} $C_r$ of $\HG_a$ such that
$C_r\subseteq \treecomp(h_r)$, there is exactly one child $h_s$ of $h_r$ such that $C_r=\treecomp(h_s)$.

Let $C_r$ be an \component{h_r} such that $C_r\subseteq \treecomp(h_r)$. Assume, first, that $h_r$ is the child of a vertex $h_p\in
\edges(\HG_a)$ of $\JT[h]$, i.e., $h_r$ is distinct from the root $h$ of $\JT[h]$. Then, because of the part {\sc subtrees$\mapsto$components}
above, we have that $\nodes(\JT[h]_{h_r})= \treecomp(h_r)\cup (h_p\cap h_r)$. In particular, this entails that $\nodes(\JT[h]_{h_r})\supseteq
\treecomp(h_r)$. Thus, $\nodes(\JT[h]_{h_r})\supseteq C_r$. Then, since $h_r\cap C_r=\emptyset$, we have that for each node $X\in C_r$, $X$
occurs in some vertex of a subtree of $\JT[h]_{h_r}$ rooted at a child $h_s(X)$ of $h_r$, with $X\in h_s(X)$. In particular, because of the
connectedness condition of join trees, there is precisely one such subtree, since $X\not\in h_r$. Now, we can apply the part {\sc
subtrees$\mapsto$components} above on $h_s(X)$ to observe that there is exactly one \component{h_r} $\treecomp(h_s)$ of $\HG_a$ such that
$\nodes(\JT[h]_{h_s(X)})= \treecomp(h_s)\cup (h_r\cap h_s(X))$. However, since $X\not\in h_r$, $X\in \treecomp(h_s)$ holds. Hence,
$C_r=\treecomp(h_s)$.

Finally, consider now the case where $h_r$ is the root of $\JT[h]$, i.e., $h_r=h$. Then, let $C_r$ be an \component{h_r} and let $X\in C_r$.
Let $h_X$ be the hyperedge that is the closest to the root of $\JT[h]$ and such that $X\in h_X$. Note that because of the connectedness
condition, there is precisely one such hyperedge $h_X$. By using the same line of reasoning as above, it follows that the child $h_s(X)$ of $h$
such that $X$ occurs in some vertex of $\JT[h]_{h_s(X)}$ is the only one satisfying the condition in the statement. \hfill $\Box$
\end{proof}

\subsection{Preservation of Components}

In the light of Theorem~\ref{thm:JTstrong}, the connectivity of an arbitrary tree projection $\HG_a$ for $\HG_1$ (with respect to some
hypergraph $\HG_2$) is characterized in terms of its components. We next show that it can be also characterized in terms of the components of
the original hypergraph $\HG_1$. This is formalized in the following two lemmas.

\begin{lemma}\label{lem:containEdges1}
Let $\HG_1$ and $\HG_a$ be two hypergraphs with the same set of nodes such that $\HG_1\leq\HG_a$. Then, for each $h\in \edges(\HG_a)$ and
\icomponent{h} $C_1$ in $\HG_1$, there is an \icomponent{h} $C_a$ of $\HG_a$ such that $C_1\subseteq C_a$.
\end{lemma}
\begin{proof}
Since $\HG_1\leq \HG_a$, for each hyperedge $h'\in \edges(\HG_1)$, there is a hyperedge $h_a\in \edges(\HG_a)$ such that $h_1\subseteq h_a$.
Then, for any set of nodes $h$ and any \component{h} $C_1$ of $\HG_1$, it follows that $C_1$ is also \connected{h} in $\HG_a$. Hence, there is
an \component{h} $C_a$ of $\HG_a$ such that $C_1\subseteq C_a$.\hfill $\Box$
\end{proof}

\begin{lemma}\label{lem:containEdges2}
Let $\HG_1$ and $\HG_a$ be two hypergraphs with the same set of nodes such that $\HG_1\leq\HG_a$. Then, for each $h\in \edges(\HG_a)$ and
\icomponent{h} $C_a$ in $\HG_a$, there are $C_1^1,...,C_1^n$ \icomponent{h}s of $\HG_1$ such that $C_a=\bigcup_{i=1}^n C_1^i$.
\end{lemma}
\begin{proof}
After Lemma~\ref{lem:containEdges1}, the result follows from the fact that $\HG_1$ and $\HG_a$ are defined over the same set of nodes. Indeed,
let $X$ be a node in $C_a$. Then, since $X\not\in h$, $X$ belongs to an \component{h} $C(X)$ of $\HG_1$, and because of
Lemma~\ref{lem:containEdges1}, $C(X)\subseteq C_a$ holds. Thus, $C_a=\bigcup_{X\in C_a} C(X)$. \hfill $\Box$
\end{proof}

At a first sight, however, since each hyperedge in $\HG_1$ is contained in a hyperedge of $\HG_a$, one may naturally be inclined at thinking
that such a ``bigger'' hypergraph $\HG_a$ is characterized by a higher connectivity, because some nodes that are not (directly) connected by
any edge in $\HG_1$ may be included together in some edge of $\HG_a$. Indeed, in general, for any given set of nodes $h$, evaluating
\component{h}s of $\HG_1$ gives proper subsets of the analogous components evaluated in $\HG_a$. Next, we show that this is not the case if
minimal tree projections are considered.

\begin{figure}[t]
  \centering
  \includegraphics[width=0.9\textwidth]{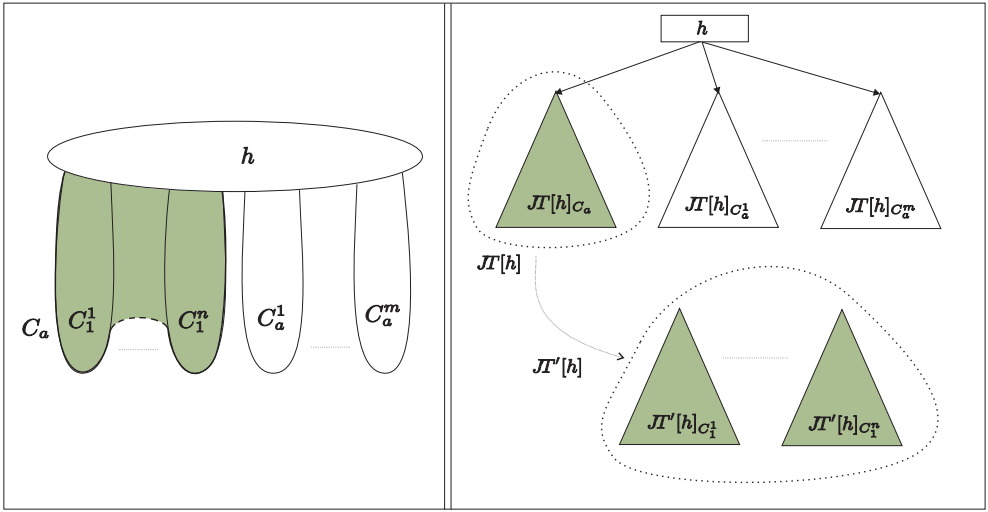}
  \caption{Illustration for the proof of Theorem~\ref{thm:normalization}.}\label{fig:nf}
\end{figure}

\begin{theorem}\label{thm:normalization}
Let $\HG_a$ be a minimal tree projection of $\HG_1$ w.r.t.~$\HG_2$. Then, for each hyperedge $h\in \edges(\HG_a)$, $C\mbox{ is an
}\icomponent{h}\mbox{ of }\HG_a \Leftrightarrow C\mbox{ is an }\icomponent{h}\mbox{ of }\HG_1$.
\end{theorem}
\begin{proof}
Let $\HG_a$ be a minimal tree projection of $\HG_1$ with respect to $\HG_2$. Let $h$ be in $\edges(\HG_a)$, and assume, by contradiction, that:
$ C\mbox{ is an }\component{h}\mbox{ of }\HG_a \not\Leftrightarrow C\mbox{ is an }\component{h}\mbox{ of }\HG_1.$
From Lemma~\ref{lem:containEdges1} and Lemma~\ref{lem:containEdges2}, it follows that there is an \component{h} $C_a$ in $\HG_a$, and
$n> 1$ \component{h}s  $C_1^1,...,C_1^n$ of $\HG_1$ such that $C_a=\bigcup_{i=1}^n C_1^i$. See Figure~\ref{fig:nf}, for an illustration.

Let $H$ be the set of all hyperedges of $\HG_a$ that intersect $C_a$, i.e., $H=\{ h_a \mid h_a\in \edges(\HG_a)\wedge h_a\cap C_a\neq
\emptyset\}$, and consider the hypergraph $\HG_a'$ defined over the same set of nodes of $\HG_a$ and such that:
$$\edges(\HG_a')  =  (\edges(\HG_a)- H)\ \cup 
                \{ h_a\cap (C_1^i\cup h) \mid h_a\in H, i\in \{1,\dots,n\}\ \}.$$

Note that, since $C_a=\bigcup_{i=1}^n C_1^i$ with $n>1$, there is at least a hyperedge $\bar h_a\in \edges(\HG_a)$ such that $\bar
h_a\cap(C_1^i\cup h)\subset \bar h_a$, for some \component{h} $C_1^i$. Thus, $\HG_a'\neq \HG_a$.
Let in fact $h_1$ be any hyperedge in $\edges(\HG_a')\setminus \edges(\HG_a)$. Then, $h_1\in \{ h_a\cap (C_1^i\cup h) \mid h_a\in H, i\in
\{1,\dots,n\}\ \}$. That is, $h_1=h_a\cap (C_1^{\bar i}\cup h)$ for some hyperedge $h_a\in H$ and \component{h} $C_1^{\bar i}$ of $\HG_1$. In
particular, note that the case where $h_1=h_a=h_a\cap (C_1^{\bar i}\cup h)$ is impossible, for otherwise we would have $h_1\in \edges(\HG_a)$.
Thus, $h_1=h_a\cap (C_1^{\bar i}\cup h)\subset h_a$, which in turn entails that $h_a\cap (C_1^i\cup h)\subset h_a$, for each \component{h}
$C_1^i$.
This property suffices to show that $h_a\not\in\edges(\HG_a')$. Indeed, assume by contradiction that $h_a\in\edges(\HG_a')$. As $h_a\in H$,
there is a hyperedge $h_a'\in H$ such that $h_a=h_a'\cap (C_1^i\cup h)$ for some \component{h} $C_1^i$, and therefore such that $h_a\subseteq
h_a'$. However, since $h_a\cap (C_1^i\cup h)\subset h_a$, we conclude that $h_a\neq h_a'$ and, hence, $h_a\subset h_a'$. This is impossible
since $\HG_a$ is a minimal tree projection, and thus a reduced hypergraph by Fact~\ref{thm:reducedMinimal2}.
It follows that $\HG_a'\subset \HG_a$, because for (the generic) hyperedge $h_1\in (\edges(\HG_a')\setminus\edges(\HG_a))$ there exists
$h_a\in(\edges(\HG_a)\setminus\edges(\HG_a'))$ such that $h_1\subset h_a$.

We now claim that the following three properties hold on $\HG_a'$.
\\

\noindent \textit{Property $P_1$:} \emph{ $\HG_a'\leq \HG_2$.}

\vspace{-2mm}
\begin{itemize}
\item[\ ] \textbf{Proof.} We have to show that for each hyperedge $h_a'\in \edges(\HG_a')$, there is a hyperedge $h_2 \in \edges(\HG_2)$
    such that $h_a'\subseteq h_2$. To this end, observe that for each hyperedge $h_a'\in \edges(\HG_a')$, there is by definition of
    $\edges(\HG_a')$ a hyperedge $h_a \in \edges(\HG_a)$ such that $h_a'\subseteq h_a$. Then, since $\HG_a$ is a tree projection of $\HG_1$
    w.r.t.~$\HG_2$, there is in turn a hyperedge $h_2\in \edges(\HG_2)$ such that $h_a\subseteq h_2$. That is, $h_a'\subseteq h_2$, for
    some $h_2\in \edges(\HG_2)$. \hfill $\diamond$
\end{itemize}

\noindent \textit{Property $P_2$:} \emph{ $\HG_1\leq \HG_a'$.}

\vspace{-2mm}\begin{itemize}
\item[\ ] \textbf{Proof.}  We have to show that for each hyperedge $h_1\in \edges(\HG_1)$, there is a hyperedge $h_a' \in \edges(\HG_a')$ such
    that $h_1\subseteq h_a'$. Let $h_1$ be a hyperedge of $\HG_1$. Since $\HG_a$ is a tree projection of $\HG_1$, we have that there is a
    hyperedge $h_a\in \edges(\HG_a)$ such that $h_1\subseteq h_a$. In the case where $h_1\cap C_a= \emptyset$, we distinguish two subcases.
    Either $h_1\subseteq h$, or $h_1\setminus h\neq \emptyset$. In the former scenario, we have just to observe that $h$ occurs in
    $\edges(\HG_a')$, as $h\cap C_a=\emptyset$, and hence $h=h_a$. In the latter scenario, $h_a\cap C_a$ must be empty, as $h_a$ is
    \connected{h} in $\HG_a$ and $h_1\subseteq h_a$. Again, we have that $h_a$ occurs in $\edges(\HG_a')$.
Consider now the case where $h_1\cap C_a\neq \emptyset$, and let $X\in h_1\cap C_a$. Because of Lemma~\ref{lem:containEdges2}, $X$ must belong
to an \component{h} $C_1^i$ in $\HG_1$. Then, $\edges(\HG_a')$ contains, by definition, the hyperedge $h_a'=h_a\cap (C_1^i\cup h)$. In fact,
since $h_1\subseteq h_a$, we also have $h_1\cap (C_1^i\cup h)\subseteq h_a'$. In order to conclude that  $h_1\subseteq h_a'$, it remains to
observe that all the vertices in $h_1\setminus h$ are contained in $C_1^i$ since $h_1\setminus h$ is \connected{h} in $\HG_1$ and $X\in h_1\cap
C_1^i$. \hfill $\diamond$
\end{itemize}

\noindent \textit{Property $P_3$:} \emph{ $\HG_a'$ is acyclic.}

\vspace{-2mm}\begin{itemize}
\item[\ ] \textbf{Proof.} The proof of this property is rather technical, and hence we find convenient to illustrate its main ideas here, as
    they shed some light on the connectivity of minimal tree projections.
From Theorem~\ref{thm:JTstrong}, we know that $\HG_a$ has an $\HG_a$-component tree rooted at $h$, say $\JT[h]$. For such a join tree, there is
a one-to-one correspondence between components of $\HG_a$ and subtrees of $\JT[h]$. Accordingly, any such a component $C$, denote by
$\JT[h]_{C}$ the subtree rooted at the child $h_s$ of $h$ such that $C=\treecomp(h_s)$.
Then, the line of the proof is to apply a normalization procedure over the subtree $\JT[h]_{C_a}$ which is in charge of decomposing $C_a$, in
order to build the subtrees $\JT'[h]_{C_1^1}$,...,$\JT'[h]_{C_1^n}$, each one being in charge of decomposing an \component{h} in $\HG_1$. An
illustration is reported in Figure~\ref{fig:nf}. The resulting tree $\JT'[h]$ can be shown to be a join tree for $\HG_a'$, thus witnessing that
$\HG_a'$ is acyclic.

Let us now prove formally the result. Recall that $\HG_a$ is reduced because of Fact~\ref{thm:reducedMinimal2}.
From Theorem~\ref{thm:JTstrong}, we know that $\HG_a$ has an $\HG_a$-component tree rooted at $h$, say $\JT[h]$. For such a join tree, there is
a one-to-one correspondence between components of $\HG_a$ and subtrees of $\JT[h]$. Accordingly, for any such a component $C$, denote by
$\JT[h]_{C}$ the subtree rooted at the child $h_s$ of $h$ such that $C=\treecomp(h_s)$.

Let $C_a,C_a^1,...,C_a^m$ be the \component{h}s of $\HG_a$, where $C_a$ is the component such that $C_a=\bigcup_{i=1}^n C_1^i$, with $n>1$ and
$C_1^1,...,C_1^n$ are \component{h}s of $\HG_1$. Based on $\JT[h]$, we shall build a tree $\JT'[h]$ whose vertices are the hyperedges of
$\HG_a'$. In particular, $\JT'[h]$ is a built as follows:
\begin{itemize}
  \item The root of $\JT'[h]$ is the hyperedge $h$.
  \item Each subtree $\JT[h]_{C_a^i}$ occurs in $\JT'[h]$ as a subtree of $h$.
  \item For each \component{h} $C_1^i\subseteq C_a$ in $\HG_1$, $\JT'[h]$ contains, as a subtree of $h$, the subtree $\JT'[h]_{C_1^i}$ that
      is built from $\JT[h]_{C_a}$ by replacing each hyperedge $h_a$ with the hyperedge $h_a\cap(C_1^i\cup h)$.
  \item No further vertices are in $\JT'[h]$.
\end{itemize}

Next, we show that $\JT'[h]$ is a join tree. Actually, $\JT'[h]$ may contain two vertices associated to the same hyperedge of $\HG_a'$ (because
of different original hyperedges that may lead to the same intersections). Thus, formally $\JT'[h]$ cannot be precisely a join tree, and we
shall rather show that it is a hypertree decomposition of width 1 where $\chi(p)=\nodes(\lambda(p))$, for each vertex $p$, which of course
entails the acyclicity of the considered hypergraph. However, for the sake of presentation, we keep the notation of join trees, avoiding the
use of the $\chi$ and $\lambda$-labelings, and we allow that $\JT'[h]$ contains two vertices associated with the same hyperedge of $\HG_a'$.

\medskip
\emph{(i) For each vertex $h'$ in $\JT'[h]$, $h'$ is in $\edges(\HG_a')$.}
Let $h'$ be in $\JT'[h]$. In the case where $h'=h$, or $h'$ occurs in a subtree of the form $\JT'[h]_{C_a^i}$, then $h'$ precisely coincides
with a hyperedge of $\HG_a$ such that $h'\not \in H$. Thus, $h'$ also belongs to $\edges(\HG_a')$, by definition. If $h'$ occurs in a subtree
of the form $\JT'[h]_{C_1^i}$, then $h'=h_a\cap (C_1^i\cup h)$, by construction of $\JT'[h]$, for some hyperedge $h_a$ in $\JT[h]_{C_a}$ which
is, hence, such that $h\neq h_a$. In particular, because of Theorem~\ref{thm:JTstrong}, $h_a\setminus h\subseteq C_a$.
The case $h_a\subseteq h$ (actually, $h_a\subset h$) is impossible, since $h_a$ and $h$ are both hyperedges of $\HG_a$, which is minimal and
hence reduced by Fact~\ref{thm:reducedMinimal2}. Thus, $h_a\cap C_a\neq \emptyset$ and hence $h_a\in H$. Then, the hyperedge $h'=h_a\cap
(C_1^i\cup h)$ is in $\edges(\HG_a')$.

\medskip
\emph{(ii) For each hyperedge $h'$ in $\edges(\HG_a')$, $h'$ is in $\JT'[h]$.}
Let $h'\neq h$ be a hyperedge of $\HG_a'$; indeed, for $h'=h$ the property trivially holds. If $h'$ is also a hyperedge of $\HG_a$, then either
$h'\not\in H$, or $h'\in H$ and there is an \component{h} $C_1^i$ with $h'=h'\cap (C_1^i\cup h)$, i.e., with $h'\subseteq (C_1^i\cup h)$. If
$h'\not\in H$, then $h'\cap C_a= \emptyset$. Then, we have that $h'\cap C_a^i\neq \emptyset$ for some \component{h} $C_a^i\neq C_a$. Hence, due
to Theorem~\ref{thm:JTstrong}, $h'$ occurs in $\JT[h]_{C_a^i}$. The result then follows since $\JT[h]_{C_a^i}$ also occurs as a subtree of
$\JT'[h]$.
Consider now the case where $h'\in H$ and there is an \component{h} $C_1^i$ with $h'=h'\cap (C_1^i\cup h)$, i.e., with $h'\subseteq (C_1^i\cup
h)$. Since $h'\not\subseteq h$, it holds that $h'\cap C_1^i\neq \emptyset$ and hence, due to Lemma~\ref{lem:containEdges2}, $h'\cap C_a\neq
\emptyset$. Then, $h'$ occurs in $\JT[h]_{C_a}$ because of Theorem~\ref{thm:JTstrong} and, by construction, $h'$ occurs in $\JT'[h]_{C_1^i}$.
Finally, assume that $h'$ is not a hyperedge of $\HG_a$. Thus, $h'=h_a \cap (C_1^i\cup h)$, for some hyperedge $h_a\in \edges(\HG_a)$ and
\component{h} $C_1^i$ with $h_a\cap C_a\neq \emptyset$ and $h_a\not\subseteq (C_1^i\cup h)$. Due to Theorem~\ref{thm:JTstrong}, $h_a$ occurs in
$\JT[h]_{C_a}$. Then, by construction, $h'$ occurs in $\JT'[h]_{C_1^i}$.

\medskip
\emph{(iii) The connectedness condition holds on $\JT'[h]$.}
Let $h_{a_1}'$ and $h_{a_2}'$ be two hyperedges in $\HG_a'$ such that $h_{a_2}'$ occurs in the subtree of $\JT'[h]$ rooted at $h_{a_1}'$. Since
subtrees of the form $\JT[h]_{C_a^i}$ are not altered in the transformation, we can focus on the case where $h_{a_2}'$ occurs in some subtree
of the form $\JT'[h]_{C_1^i}$ and where either $h_{a_1}'=h$ or $h_{a_1}'$ occurs in the same subtree. In fact, $h_{a_2}'$ (resp., $h_{a_1}'$)
belonging to $\JT'[h]_{C_1^i}$ entails that $h_{a_2}'=h_{a_2}\cap(C_1^i\cup h)$ (resp., $h_{a_1}'=h_{a_1}\cap(C_1^i\cup h)$), for some
hyperedge $h_{a_2}\in \edges(\HG_a)$ (resp., $h_{a_1}\in \edges(\HG_a)$). Note that to deal uniformly with the two cases above, if
$h_{a_1}'=h$, then we can just set $h_{a_1}=h$.
Now, let $Y$ be a node in $h_{a_1}'\cap h_{a_2}'$. Then, $Y$ belongs to $h_{a_2}\cap h_{a_1}$. Consider a hyperedge $h_a'$ in the path between
$h_{a_1}'$ and $h_{a_2}'$. Again, $h_a'$ belonging to $\JT'[h]_{C_1^i}$ entails that $h_a'=h_a \cap(C_1^i\cup h)$, where $h_a$ is an edge
occurring in the path between $h_{a_1}$ and $h_{a_2}$ in $\JT[h]$. Since $\JT[h]$ is a join tree, $Y$ also occurs in $h_a$, and hence $Y$ is in
$h_a'$.~\hfill~$\diamond$
\end{itemize}

In the light of the above properties, $\HG_a'$ is a tree projection of $\HG_1$ w.r.t.~$\HG_2$ such that $\HG_a'\subset \HG_a$, thereby
contradicting the fact that $\HG_a$ is a minimal tree projection. \hfill $\Box$
\end{proof}

\subsection{Connected Tree Projections}

We next present another interesting property of minimal tree projections: they always admit join trees in a desirable form that we call
connected. Such a form is based on  the well known notion of connected decomposition defined for the treewidth (see, e.g.,~\cite{FN06}). Let
$\tuple{T,\chi}$ be a tree decomposition of a graph $G$. For any pair of adjacent vertices $p_r$ and $p_s$ of $T$, let $T_r$ and $T_s$ be the
two connected subtrees obtained from $T$ by removing the edge connecting $p_r$ and $p_s$. Then, $\tuple{T,\chi}$ is {\em connected} if the
sub-graphs induced by the nodes covered be the $\chi$-labeling in $T_r$ and in $T_s$, respectively, are connected, for each pair of vertices
$p_r$ and $p_s$.

Next, we define a natural extension of this notion to the more general framework of tree projections of hypergraph pairs.

\begin{definition}\label{def:connected-TP}\em
A tree projection $\HG_a$ of $\HG_1$ w.r.t.~$\HG_2$ is \emph{connected} if it has an $\HG_1\mbox{-}{\it connected}$ join tree, i.e., a join
tree $\JT$ with the following property: For each pair of adjacent vertices $h_r,h_s$ of $\JT$, the sub-hypergraph of $\HG_1$ induced by the
nodes in $\nodes(\JT[h_r]_{h_s})$ is \connected{\emptyset}.~\hfill~$\Box$
\end{definition}

\begin{figure*}[t]
  \centering
  \includegraphics[width=0.7\textwidth]{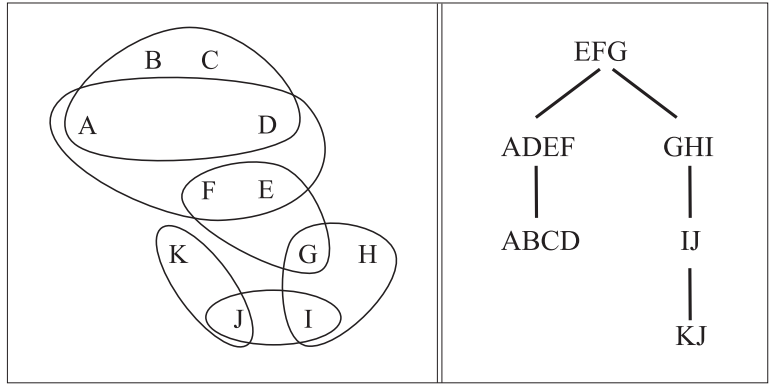}
  \caption{A minimal tree projection $\HG_a'$ of $\HG_{1}$ w.r.t.~$\HG_{2}$, and a join tree in normal form $\JT_a'$ for it.}\label{fig:minimalTP}
\end{figure*}

Note that the novel notion coincides with the original one whenever we considers the treewidth method, that is, whenever we look for tree
projections of pairs of the form $(\HG,\HG^{tk})$, for any fixed natural number $k>0$, where $\HG^{tk}$ is the hypergraph whose hyperedges are
all possible sets of at most $k+1$ nodes in $\HG$.

\begin{example}\em
The tree projection $\HG_a$ of $\HG_{1}$ w.r.t.~$\HG_{2}$ reported in Figure~\ref{fig:hypegraph} is not connected, because it has no
$\HG_{1}$-connected join trees.

For instance, consider the join tree $\JT_a$ depicted in the same figure, and let $h_r=\{E,F,G,H,I,J,K\}$ and $h_s=\{A,D,E,F,J,K\}$. Then, the
sub-hypergraph of $\HG_{1}$ induced by $\nodes(\JT_a[h_r]_{h_s})$ consists of the hyperedges $\{D,F,E\}$, $\{K,J\}$, $\{A,B,C\}$, $\{C,D\}$,
and $\{A,F\}$, and thus $\{K,J\}$ is clearly disconnected from the others. On the other hand, note that the join tree $\JT_a'$ for the minimal
tree projection $\HG_a'$ reported in Figure~\ref{fig:minimalTP} is $\HG_{1}$-connected.  \hfill $\lhd$
\end{example}

We next show that such a connected join tree always exists for any minimal tree projection, as in the special case of the tree decomposition
method. In order to establish the result, we shall exploit an algorithm, called {\em make-it-connected}, that has been described in~\cite{FN06}
and that enjoys the following properties.

\begin{proposition}[cf. \cite{FN06}]\label{prop:makeConnect} Let $\tuple{T,\chi}$ be a width-$k$ tree decomposition of a graph $G$. Then,  {\em Algorithm~make-it-connected}
builds in polynomial time a connected width-$k'$ tree decomposition $\tuple{T',\chi'}$ of $G$, with $k'\leq k$, such that: {\em (1)} for each
vertex $p'$ of $T'$, there is a vertex $p$ of $T$ such that $\chi'(p')\subseteq \chi(p)$; and {\em (2)} if $\tuple{T,\chi}$ is not connected,
then there is a vertex $\bar p$ of $T$ such that $\chi(\bar p)\neq \chi'(p')$, for each vertex $p'$ of $T'$.
\end{proposition}

\begin{theorem}\label{thm:connected}
If $\HG_a$ is a minimal tree projection of $\HG_1$ w.r.t.~$\HG_2$, then any join tree $\JT$ for $\HG_a$ is $\HG_1$-connected.
\end{theorem}

\begin{proof}
Assume that $\HG_a$ is a minimal tree projection of $\HG_1$ w.r.t.~$\HG_2$, hence from Fact~\ref{fact:sameNodes} $\nodes(\HG_1)=\nodes(\HG_a)$.
Let $\JT$ be a join tree for $\HG_a$, and let $\tuple{T,\chi}$ be a labeled tree whose vertices one-to-one correspond with the vertices of
$\JT$. In particular, for each hyperedge $h\in\edges(\HG_a)$, $T$ contains the vertex $p_h$, which is moreover such that $\chi(p_h)=h$.
From the connectedness property of join trees and the fact that $\HG_1\leq \HG_a$, it immediately follows that $\tuple{T,\chi}$ is a tree
decomposition of (the Gaifman graph of) $\HG_1$.
Assume now, for the sake of contradiction, that $\JT$ is not $\HG_1$-connected. Then, $\tuple{T,\chi}$ is not connected too. Thus, we can apply
algorithm {\em make-it-connected} on $\tuple{T,\chi}$, which produces the connected tree decomposition $\tuple{T',\chi'}$, with $T'=(N',E')$ of
(the Gaifman graph of) $\HG_1$.

Let $\HG_a'$ be the acyclic hypergraph such that $\nodes(\HG'_a)=\nodes(\HG_a)$ and $\edges(\HG_a')= \{ \chi'(p') \mid p'\in N'\}$, and let
$\HG_a''$ be the reduced hypergraph obtained from $\HG_a'$ by removing its hyperedges that are proper subsets of some hyperedge in $\HG_a'$.
Therefore, we have $\edges(\HG_a'')\subseteq \edges(\HG_a')$ and $\HG_a'\leq \HG_a''$. Of course, $\HG_a''$ is acyclic too. Moreover, we claim
that $\HG_a''\subseteq \HG_a$. Indeed, for each hyperedge $\chi'(p')\in \edges(\HG_a'')\setminus \edges(\HG_a)$, by
Proposition~\ref{prop:makeConnect}.(1), there is a hyperedge $\chi(p)\in \edges(\HG_a)$ such that $\chi'(p')\subseteq \chi(p)$. Moreover,
$\chi(p)$ cannot occur in $\edges(\HG_a'')$, as $\HG_a''$ is reduced. Hence, $\chi(p)$ is in $\edges(\HG_a)\setminus \edges(\HG_a'')$, and we
actually have $\chi'(p')\subset \chi(p)$. That is, $\HG_a''$ is an acyclic hypergraph with $\HG_a''\subseteq \HG_a$.

Now, observe that since $\tuple{T',\chi'}$ is a tree decomposition of (the Gaifman graph of) $\HG_1$ and since $\HG_a''\subseteq \HG_a$, we
have
$\HG_1 \leq \HG_a'\leq\HG_a'' \leq \HG_a \leq \HG_2$. Thus, $\HG_a''$ is a tree projection of $\HG_1$ w.r.t.~$\HG_2$.
However, by Proposition~\ref{prop:makeConnect}.(2), there is a vertex $\bar p$ of $T$ such that $\chi(\bar p)\neq \chi'(p')$, for each vertex
$p'$ of $T'$. Thus, $\HG_a''\neq \HG_a$. Hence, $\HG_a''$ is a tree projection for $(\HG_1,\HG_2)$ such that $\HG_a'' \subset \HG_a$, which
contradicts the minimality of $\HG_a$. \hfill $\Box$
\end{proof}

Eventually, by exploiting Fact~\ref{thm:reducedMinimal}, we get the following corollary.

\begin{corollary}\label{cor:connected}
$(\HG_1,\HG_2)$ has a tree projection if, and only if, $(\HG_1,\HG_2)$ has a connected tree projection.
\end{corollary}

\begin{remark}\label{rem:alternative-connectedness}\em
A different notion of connected decomposition has been introduced in~\cite{SH07} for the special case  of (generalized) hypertree
decompositions, in order to speed-up their computation. According to~\cite{SH07}, a (generalized) hypertree decomposition
$\HD=\tuple{T,\chi,\lambda}$ is {connected} if the root $r$ of $T$ is such that $|\lambda(r)|=1$, and for each pair of nodes $p$ and $s$, with
$s$ child of $p$ in $T$, and for each $h\in \lambda(s)$, $h\cap \chi(s)\cap\chi(p)\neq \emptyset$. The connected (generalized) hypertree width
$c(g)hw$ is the minimum width over all the possible connected (generalized) hypertree decompositions.
Whether or not $chw(\HG)=hw(\HG)$ for every hypergraph $\HG$ was an open question~\cite{SH07}.

Next, we give a negative answer to this question by showing that the latter notion of connectedness gives a structural method that is weaker
than the unrestricted (generalized) hypertree decomposition, even on graphs.
Consider the graph $G_{hex}$ in Figure~\ref{fig:controesempio}. As shown in the same figure, there is a hypertree decomposition
$\HD_{hex}=\tuple{T,\chi,\lambda}$ of this (hyper)graph having width 3, and thus $hw(G_{hex})\leq 3$. In $\HD_{hex}$, for each vertex $p$ of
$T$, $\chi(p)=\nodes(\lambda(p))$ holds, and thus Figure~\ref{fig:controesempio} shows only the $\lambda$-labeling of each vertex. Moreover,
only the left branch is detailed, showing how to deal with the upper cluster of hexagons. The other subtrees are of the same form, and thus are
not reported, for the sake of simplicity.
Note that $\{0\}$ and its child $\{0,21,42\}$ violate the required connectedness property. In fact, it turns out that the only way to attack
such hexagons is by using, at some vertex $s$ of the decomposition tree, some nodes that are not directly connected to (the nodes occurring in)
the parent vertex of $s$. Indeed, the reader can check there is neither a hypertree decomposition nor a generalized hypertree decomposition of
$G_{hex}$ that is connected according to~\cite{SH07} and has width 3. Thus, the following holds.
\end{remark}

\begin{figure}[t]
  \centering
  \includegraphics[width=0.99\textwidth]{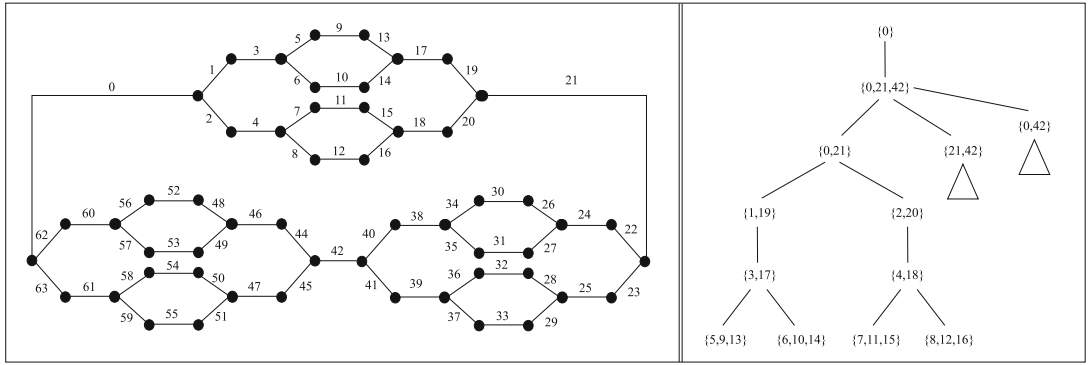}
  \caption{A graph $G_{hex}$ with $cghw(G_{hex})=4>3=hw(G_{hex})$.}\label{fig:controesempio}
\end{figure}

\begin{fact}\label{thm:connected2}
There is a graph $G_{hex}$ such that $cghw(G_{hex}) > hw(G_{hex})$.
\end{fact}

\subsection{Tree Projections in Normal Form}

We next show the main result of this section, where all the above ingredients are exploited together: all minimal tree projections have join
trees in a suitable normal form. This normal form is of theoretical interest, since it can be exploited to establish further results on the
setting (as its game-theoretic characterization discussed in Section~\ref{GTC}). Moreover, it is of practical interest, since it can be used to
prune the search space in solution approaches aimed at computing tree projections.

\begin{definition}\label{def:nf}\em
A join tree of a tree projection of $\HG_1$ w.r.t.~$\HG_2$ is said in {\em normal form} if it is $\HG_1$-connected and it is an
$\HG_1$-component tree. \hfill $\Box$
\end{definition}

\begin{example}\em
Consider again the tree projection $\HG_a$ and its join tree $\JT_a$ illustrated in Figure~\ref{fig:hypegraph}. Consider the vertices
$h_r=\{E,F,G,H,I,J,K\}$ and $h_s=\{A,D,E,F,J,K\}$ in $\JT_a[\{E,F,G,H,I,J,K\}]$, and note that there is exactly one \component{h_r}
$\treecomp(h_s)=\{A,B,C,D\}$ of $\HG_{1}$ such that $\nodes(\JT[h]_{h_s})=\{A,B,C,D,E,F,J,K\}=\treecomp(h_s)\cup (h_s\cap h_r)$. However,
$\F(\treecomp(h_s),\HG_{1})=\{A,B,C,D,E,F\}$ and hence $h_s\not\subseteq \F(\treecomp(h_s),\HG_{1})$. Thus, one condition in the ({\sc
subtrees$\mapsto$components})-part of Theorem~\ref{thm:mainMinimal} is violated.

Indeed, the tree projection $\HG_a$ is not minimal. This is witnessed by the tree projection $\HG_a'$ for $(\HG_{1},\HG_{2})$ that is reported
on the left of Figure~\ref{fig:minimalTP} and that is properly contained in $\HG_a$. A join tree $\JT_a'$ for $\HG_a'$ is reported on the right
of the same figure. The careful reader may check that $\JT_a'$ satisfies all conditions in Theorem~\ref{thm:mainMinimal}. \hfill $\lhd$
\end{example}

\begin{theorem}[Normal Form]\label{thm:mainMinimal}
Let $\HG_a$ be a minimal tree projection of $\HG_1$ w.r.t.~$\HG_2$. For any hyperedge $h\in\edges(\HG_a)$, there is a join tree for $\HG_a$ in
normal form rooted at $h$.
\end{theorem}
\begin{proof}
Let $h\in\edges(\HG_a)$ be any hyperedge of the tree projection $\HG_a$. Since minimal tree projections are reduced, Theorem~\ref{thm:JTstrong}
entails that $\HG_a$ has a join tree $\JT$ that is an $\HG_a$-component tree rooted at $h$. Again from minimality and
Theorem~\ref{thm:normalization}, we get that $\component{h'}$s in $\HG_a$ and  $\component{h'}$s in $\HG_1$ do coincide, for each
$h'\in\edges(\HG_a)$. Thus, $\JT$ has the following properties:

\vspace{-1mm}
\begin{description}
  \item[\sc subtrees$\mapsto$components.] For each vertex $h_r$ of $\JT[h]$ and each child $h_s$ of $h_r$, there is exactly one
      \component{h_r} $\treecomp(h_s)$ of $\HG_1$ such that $\nodes(\JT[h]_{h_s})=\treecomp(h_s)\cup (h_s\cap h_r)$. Moreover, $h_s\cap
      \treecomp(h_s)\neq \emptyset$ holds.

  \item[\sc components$\mapsto$subtrees.] For each vertex $h_r$ of $\JT[h]$ and each \component{h_r} $C_r$ of $\HG_1$ such that
      $C_r\subseteq \treecomp(h_r)$,  there is exactly one child $h_s$ of $h_r$ such that $C_r=\treecomp(h_s)$.
\end{description}

Hence, in order to prove that $\JT$ is an $\HG_1$-component tree, it remains to show that, for each vertex $h_r$ of $\JT[h]$ and each child
$h_s$ of $h_r$, $h_s\subseteq \F(\treecomp(h_s),\HG_1)$ holds.
Assume, for the sake of contradiction, that there is a vertex $h_r$ and a child $h_s$ of $h_r$ such that $h_s\subseteq
\F(\treecomp(h_s),\HG_1)$ does not hold. From Theorem~\ref{thm:JTstrong}, we know that $h_s\subseteq \F(\treecomp(h_s),\HG_a)$ holds. It
follows that there exists a non-empty set $W\subseteq h_s\setminus \treecomp(h_s)$ of nodes such that $X\not\in \F(\treecomp(h_s),\HG_1)$, for
each $X\in W$. Moreover, as $\nodes(\JT[h]_{h_s})=\treecomp(h_s)\cup (h_s\cap h_r)$, we have that $W\subseteq h_s\cap h_r$. Consider the
hypergraph $\HG_a'$ obtained from $\HG_a$ by replacing each hyperedge $\bar h$ occurring in $\JT[h]_{h_s}$ with $\bar h\setminus W$, and note
that $\HG_a'\subset \HG_a$. Of course, the tree $\JT'$ obtained from $\JT$ by replacing any such $\bar h$ with $\bar h\setminus W$ is a join
tree for $\HG_a'$. Finally, $\HG_a'$ is again a tree projection for $(\HG_1,\HG_2)$ because every hyperedge of $\HG_1$ is still covered by some
vertex in $\JT'$. Indeed, there is no hyperedge $h\in\edges(\HG_1)$ such that both $h\cap W\neq\emptyset$ and $h\cap
\treecomp(h_s)\neq\emptyset$, by construction of $W$. This contradicts the fact that $\HG_a$
 is a minimal tree projection for
$(\HG_1,\HG_2)$.

Finally, from Theorem~\ref{thm:connected}, $\JT$ is $\HG_1$-connected.\hfill $\Box$
\end{proof}

\section{Game-Theoretic Characterization}\label{GTC}

The \emph{Robber and Captain} game is played on a pair of hypergraphs $(\HG_1,\HG_2)$ 
by a Robber and a Captain controlling some squads of cops, in charge of the surveillance of a number of strategic targets. The Robber stands on
a node and can run at great speed along the edges of $\HG_1$; however, she is not permitted to run trough a node that is controlled by a cop.
Each move of the Captain involves one squad of cops, which is encoded as a hyperedge $h\in \edges(\HG_2)$. The Captain may ask any cops in the
squad $h$ to run in action, as long as they occupy nodes that are currently reachable by the Robber, thereby blocking an escape path for the
Robber. Thus, ``second-lines'' cops cannot be activated by the Captain. Note that the Robber is fast and may see cops that are entering in
action. Therefore, while cops move, the Robber may run trough those positions that are left by cops or not yet occupied. The goal of the
Captain is to place a cop on the node occupied by the Robber, while the Robber tries to avoid her capture.

For a comparison, observe that this game is somehow in the middle between the Robber and Marshals game of \cite{gott-etal-03}, where the
marshals occupy a full hyperedge at each move, and the Robber and Cops game of \cite{ST93}, where each cop stands on a vertex and thus, if
there are enough cops, any subset of any edge can be blocked at each move. Instead, the Captain cannot employ ``second-lines'' cops, but only
cops whose positions are under possible Robber attacks.

\begin{figure*}[t!]
  \centering
  \includegraphics[width=0.9\textwidth]{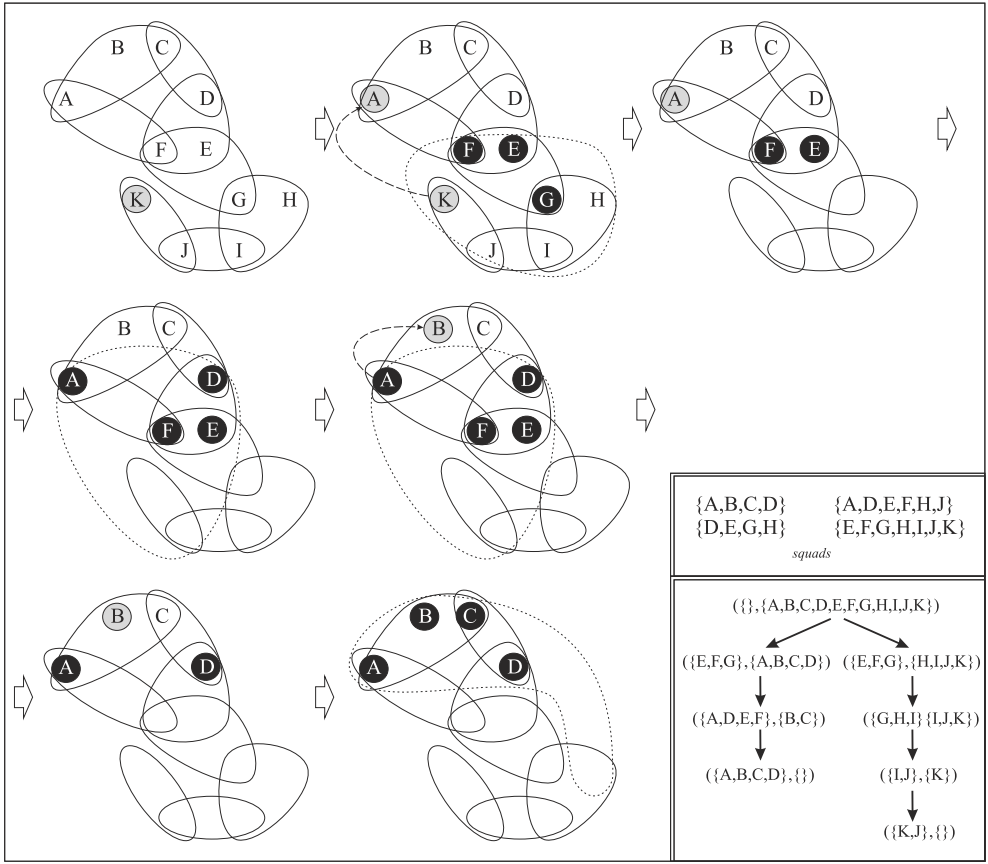}
  \caption{The Robber and Captain game on $(\HG_{1},\HG_{2})$.}\label{fig:strategy}
\end{figure*}

\begin{example}\label{ex:CR}\em
Consider the Robber and Captain game played on the pair $(\HG_{1},\HG_{2})$ of hypergraphs depicted in Figure~\ref{fig:hypegraph}, and the
sequences of moves illustrated in Figure~\ref{fig:strategy}.

Initially, the Robber stands on the node $K$, and each other node is reachable. The Captain selects the squad $\{E,F,G,H,I,J,K\}$ and uses the
three cops blocking $E$, $F$, and $G$. The Robber sees the cops and, while they enter in action, is fast enough to run on $A$. Note that, when
the Robber is on $A$ and nodes $E$, $F$, and $G$ are blocked by the Captain, the Robber can move over $\{A,B,C,D\}$, while $\{E,F\}$ are also
under possible Robber attacks because they are adjacent to her escape space. All other nodes are no longer reachable by the Robber and no
longer depicted. Hence, the Captain might ask cops to occupy some of the nodes in $\{A,B,C,D,E,F\}$, provided they are covered by some
hyperedge. In fact, the strategy of the Captain is to select the hyperedge/squad $\{A,D,E,F,J,K\}$, and then to use those cops in this squad
that block nodes $A$, $D$, $E$, and $F$. During this move of the cops, the potential escape door $\{E,F\}$ for the Robber is still blocked, and
hence its available space shrinks. Indeed, during the move of the Captain, the Robber can just move either on $B$ or on $C$. Finally, the
Captain uses the squad $\{A,B,C,D,H\}$ and order its cops to move to $A$, $B$, $C$, and $D$, thereby capturing the Robber, as its potential
escape door $\{A, D\}$ remains blocked by the cops. \hfill $\lhd$
\end{example}

In the rest of the section, we formalize and analyze the game. To this end, we intensively use the notions and the notations given in the previous section,
by implicitly applying them to the hypergraph $\HG_1$, unless stated otherwise.

\begin{definition}[\CR\ Game]\em  Let $\HG_1$ and $\HG_2$ be two hypergraphs.
The Robber and Captain game on $(\HG_1,\HG_2)$ (short: $\CR(\HG_1,\HG_2)$ game) is formalized as follows.
A \emph{position} for the Captain is a set $M$ of vertices where the cops stand such that $M\subseteq h_2$, for some hyperedge (squad) $h_2\in
\edges(\HG_2)$.
A \emph{configuration} is a pair $(M,v)$, where $M$ is a position for the Captain, and $v\in \nodes(\HG_1)$ is the node where the Robber
stands. The initial configuration is $(\{\},v_0)$, where $v_0$ is a node arbitrarily picked by the Robber.

Let $(M_i,v_i)$ be the configuration at step $i$. This is a capture configuration, where the Captain wins, if $v_i\in M_i$. Otherwise, the
Captain activates the cops in a novel position $M_{i+1}$ such that: $\forall X\in M_{i+1}$, $X$ \touches{M_i} $v_i$ (in $\HG_1$). Then, the
Robber selects some available node $v_{i+1}$ (if any) such that there is a \spath{M_i\cap M_{i+1}} from $v_i$ to $v_{i+1}$ (in $\HG_1$). If the
game continues forever, the Robber wins. \hfill $\Box$
\end{definition}

Note that it does not make sense for the Captain to assume that the Robber is on a particular node, given the ability of the Robber of changing
positions before the cops land. Thus, given a configuration $(M_i,v_i)$, we may assume w.l.o.g. that the next Captain's move is only determined
by the \component{M_i} (of $\HG_1$) that contains $v_i$, rather than by $v_i$ itself. And, accordingly, positions can equivalently be written
as $(M_i,C_i)$, where $C_i$ is an \component{M_i}. In this case, capture configurations have the form $(M,\{\})$, and the initial configuration
has the form $(\{\},\nodes(\HG_1))$.

In the following, assume that a $\CR(\HG_1,\HG_2)$ game is given. Moreover recall that, for any component $C$ (of $\HG_1$),  $\F(C,\HG_1)$ is
the {frontier} of $C$, i.e., (by omitting hereafter $\HG_1$, which is understood) the set $\F(C)=\nodes(\edges(C))=C \cup \{ Z \mid \exists
X\in C, h\in \edges(\HG_1)\mbox{ s.t. } \{X,Z\}\subseteq h \}$.
Then, observe that the moves of the Captain are confined in the frontier of the current component where the Robber stands.

\begin{fact}\label{lem:main} Let $M_i$ and $M_{i+1}$ be positions for the Captain and let $C_i$ be
an \icomponent{M_i}. Then, $\forall X\in M_{i+1}$, $X$ \itouches{M_i} $C_i$ if, and only if, $M_{i+1}\subseteq \F(C_i)$.
\end{fact}
\begin{proof}
In fact, a node $X$ is in $\F(C_i)$ if, and only if, either $X\in C_i$, or $X\in M_i$ and there is a node $Z\in C_i$ with $\{X,Z\}\subseteq h$,
for some edge $h$ in $\edges(\HG_1)$. But, this condition precisely coincides with the definition that $X$ \touches{M_i} $C_i$.\hfill $\Box$
\end{proof}

\begin{definition}[Strategies]\label{def:strategy}\em
A \emph{strategy} $\sigma$ for $\CR(\HG_1,\HG_2)$ is a function that encodes the \emph{moves} of the Captain, i.e., given a configuration
$(M_i,C_i)$, with $C_i\neq \emptyset$, $\sigma$ returns a position $M_{i+1}$ such that $M_{i+1}\subseteq \F(C_i)$.

A \emph{game-tree} for $\sigma$ is a rooted tree $T(\sigma)$ defined over configurations as follows. Its root is the configuration
$(\emptyset,\nodes(\HG_1))$. Let $(M_{i},C_{i})$ be a vertex in $T(\sigma)$ and let $M_{i+1}=\sigma(M_{i},C_{i})$.\footnote{Note the little
abuse of notation: $\sigma(M_{i},C_{i})$ instead of $\sigma((M_{i},C_{i}))$.} Then, $(M_{i},C_{i})$ has exactly one child $(M_{i+1},C_{i+1})$,
for each
\component{M_{i+1}} $C_{i+1}$ such that $C_{i}\cup C_{i+1}$ is \connected{M_{i}\cap M_{i+1}}; we 
call such a $C_{i+1}$ an \ecomponent{(M_i,C_i),M_{i+1}} for the Robber.
If there is no \ecomponent{(M_i,C_i),M_{i+1}}, then $(M_{i},C_{i})$ has exactly one child $(M_{i+1},\{\})$.
No further edge or vertex is in $T(\sigma)$.

Then, $\sigma$ is a \emph{winning} strategy if $T(\sigma)$ is a finite tree.
Moreover, define a position $M_{i+1}$ to be a \emph{monotone} move of the Captain in $(M_i,C_i)$, if for each \ecomponent{(M_i,C_i),M_{i+1}}
$C_{i+1}$, $C_{i+1}\subseteq C_i$. We say that $\sigma$ is a \emph{monotone} strategy if, for each edge from $(M_{i},C_{i})$ to
$(M_{i+1},C_{i+1})$, it holds that $M_{i+1}$ is a monotone move in $(M_i,C_i)$.~\hfill~$\Box$
\end{definition}

\begin{example}\em
Consider again the exemplification of the Robber and Captain in Figure~\ref{fig:strategy}. In particular, the bottom-right part of the figure
depicts a game-tree associated with a winning strategy: The Captain initially moves on $\{E,F,G\}$, and there are two connected components
available to the Robber, namely $\{A,B,C,D\}$ and $\{H,I,J,K\}$. The left branch of the tree illustrates the strategy when the Robber goes into
the component $\{A,B,C,D\}$. Note that this branch precisely corresponds to the moves that are discussed in Example~\ref{ex:CR}. The right
branch addresses the case where the Robber goes into the component $\{H,I,J,K\}$. In both cases, the Captain will eventually capture the
Robber. Observe that this winning strategy is monotone.\hfill $\lhd$
\end{example}

\subsection{Monotone vs Non-monotone Strategies}\label{sec:monotoneVSnonmonotone}

In this section, we show that there is no incentive for the Captain to play a strategy $\sigma$ that is not monotone, since it is always
possible to construct and play a monotone strategy $\sigma'$ that is equivalent to $\sigma$, i.e., such that $\sigma'$ is winning if, and only
if, $\sigma$ is winning. This crucial property conceptually relates our game with the \emph{Robber and Cops game} characterizing the
\emph{treewidth}~\cite{ST93}, and differentiates it from most of the hypergraph-based games in the literature, in particular, from the
\emph{Robber and Marshals game}, whose monotone strategies characterize hypertree decompositions~\cite{gott-etal-03}, while non-monotone
strategies do not correspond to valid decompositions \cite{adler04}.

We point out that the proof below does not apply to the traditional Robber and Cops game, because in our setting cops can be placed just on the
positions that are reachable by the Robber. As a matter of fact, our techniques are substantially different from those used to show that
non-monotonic moves provide no extra-power in the Robber and Cops game. We start by illustrating some properties of the novel game.

In the following, assume that $\sigma$ and $T(\sigma)$ are a strategy and a game tree for it, respectively. Moreover recall that, for any
component $C$ (of $\HG_1$), $\partial (C)$ denotes the {\em border} of $C$ (in $\HG_1$), i.e., the set $\F(C)\setminus C$.
Then, let the \emph{escape-door} of the Robber in $v_i=(M_i,C_i)$ when attacked with $M_{i+1}$ be defined as $\E(v_i,M_{i+1})=\partial
(C_i)\setminus M_{i+1}$. Note that this is equivalent to state that $\E(v_i,M_{i+1})=M_i\cap \F(C_i) \setminus M_{i+1}$, because $C_i$ is an
\component{M_i}.
Consider for instance Example~\ref{ex:CR} at the configuration $v_1=(\{E,F,G\},\{A,B,C,D\})$, when the Robber is attacked by the Captain with
the cops $\{A,D,E,F\}$. In this case, the frontier is $\F(\{A,B,C,D\})= \{A,B,C,D,E,F\}$, hence the \emph{escape-door} is $\{E,F\}\setminus
\{A,D,E,F\}=\emptyset$.

In the following lemma, we show that this set precisely characterizes those vertices trough which the Robber may escape from
the current component $C_i$, when the Captain changes her position from $M_i$ to $M_{i+1}$.

\begin{lemma}\label{lem:escapeSpace}
Let $M_i$ and $M_{i+1}$ be positions for the Captain, let $C_i$ be an \icomponent{M_i}, and let $v_i=(M_i,C_i)$. Then, $C_{i+1}$ is a
\iecomponent{v_i,M_{i+1}} if, and only if, $C_{i+1}$ is an \icomponent{M_{i+1}} with $C_{i+1}\cap (C_i \cup \E(v_i,M_{i+1}))\neq \emptyset$.
\end{lemma}

\begin{proof}
Recall from Definition~\ref{def:strategy} that $C_{i+1}$ is an \ecomponent{(M_i,C_i),M_{i+1}} if $C_{i+1}$ is an \component{M_{i+1}} such that
$C_{i+1}\cup C_i$ is \connected{M_{i+1}\cap M_i}.

{\em (if-part)} Assume that $C_{i+1}$ is an \component{M_{i+1}} with $C_{i+1}\cap (C_i \cup \E(v_i,M_{i+1}))\neq \emptyset$. Since $C_{i+1}$
(resp., $C_i$) is an \component{M_{i+1}} (resp., \component{M_i}), we have that $C_{i+1}$ (resp., $C_i$) is contained in an
\component{M_{i+1}\cap M_i}, say $C_{i+1}'$ (resp., $C_i'$). Therefore, if $C_{i+1}\cap C_i\neq \emptyset$, we immediately can conclude that
$C_{i+1}\cup C_i$ is \connected{M_{i+1}\cap M_i}, with $C_{i+1}'=C_i'$. Thus, let us consider the case where $C_{i+1}\cap C_i=\emptyset$ and,
hence, $C_{i+1}\cap \E(v_i,M_{i+1})\neq \emptyset$. Consider now $p_{i+1}\in C_{i+1}\cap \E(v_i,M_{i+1})$. By definition of $\E(v_i,M_{i+1})$,
$p_{i+1}$ belongs in particular to $M_i\cap \F(C_i)\setminus M_{i+1}$. Thus, $p_{i+1}\not\in M_{i+1}\cap M_i$. However, $p_{i+1}\in C_{i+1}$
and $p_{i+1}\in \F(C_i)$. From the latter, we have that there is a node $Z\in C_i$ and a hyperedge $h_1\in \edges(\HG_1)$ such that
$\{p_{i+1},Z\}\subseteq h_1$. It follows that there is an \spath{M_{i+1}\cap M_i} from $p_{i+1}\in C_{i+1}$ to $Z\in C_i$. Thus, $C_{i+1}\cup
C_i$ is \connected{M_{i+1}\cap M_i}.

{\em (only-if-part)} Assume that $C_{i+1}$ is an \component{M_{i+1}} such that $C_{i+1}\cup C_i$ is \connected{M_{i+1}\cap M_i}. Consider the
case where $C_{i+1}\cap C_i=\emptyset$. Then, there is a node $p_{i+1}\in M_i\cap \F(C_i)$ such that $p_{i+1}\in C_{i+1}$. Thus,
$p_{i+1}\not\in M_{i+1}$. It follows that $p_{i+1}\in M_i\cap \F(C_i)\setminus M_{i+1}$, and hence $C_{i+1}\cap \E(v_i,M_{i+1})\neq \emptyset$.
\hfill $\Box$
\end{proof}

Moreover, we next characterize monotone moves based on escape-doors.

\begin{lemma}\label{lem:monotone}
Let $M_i$ and $M_{i+1}$ be positions for the Captain, let $C_i$ be an \icomponent{M_i}, and let $v_i=(M_i,C_i)$. Then,
$\E(v_i,M_{i+1})=\emptyset$ if, and only if, for each \iecomponent{v_i,M_{i+1}} $C_{i+1}$, $C_{i+1}\subseteq C_i$.
\end{lemma}

\begin{proof}
{\em (if-part)} Assume that $\forall$ \ecomponent{v_i,M_{i+1}} $C_{i+1}$, $C_{i+1}\subseteq C_i$. Moreover, assume for the sake of
contradiction that $\E(v_i,M_{i+1})\neq \emptyset$, and let $X\in \E(v_i,M_{i+1})=M_i\cap \F(C_i)\setminus M_{i+1}$. In particular note that
$X\not\in M_{i+1}$, from which we conclude that there must be an \component{M_{i+1}} $C_{i+1}$ such that $X\in C_{i+1}$. Thus, $X\in
C_{i+1}\cap \E(r,M_{i+1})$ and hence we can apply Lemma~\ref{lem:escapeSpace} to conclude that $C_{i+1}$ is a \ecomponent{v_i,M_{i+1}}.
However, $X$ is not in $C_i$, since $X$ belongs to $M_i$ (and $C_i$ is an \component{M_{i}}). Thus, $C_{i+1}\not\subseteq C_i$, which is
impossible.

{\em (only-if-part)} Assume that $\E(v_i,M_{i+1})= \emptyset$, and for the sake of contradiction that $C_{i+1}$ is a \ecomponent{v_i,M_{i+1}}
such that $C_{i+1}\not\subseteq C_i$. Let $Y$ be a node in $C_{i+1}\setminus C_i$, and observe that there must be a node $X\in C_{i+1}\cap
C_i$, because of Lemma~\ref{lem:escapeSpace}.
Consider now an \spath{M_{i+1}} from $Y$ to $X$ and let $Z_1,Z_2$ be two nodes in this path such that $Z_1\in C_{i+1}\cap C_i$, $Z_2\in
C_{i+1}\setminus C_i$, and $\{Z_1,Z_2\}\subseteq h$ for some hyperedge $h\in \edges(\HG_1)$. Note that these two nodes exist because of the
properties of the endpoints $Y$ and $X$. Now, it must be the case that $Z_2$ is in $(\F(C_i)\setminus C_i)\cap C_{i+1}$. Since $M_i\supseteq
\F(C_i)\setminus C_i$, the latter entails that $Z_2\in M_i\cap\F(C_i)\cap C_{i+1}$. Finally, since $C_{i+1}$ is an \component{M_{i+1}}, we
conclude that $Z_2\in M_i\cap\F(C_i)\setminus M_{i+1}$, i.e., $Z_2\in \E(v_i,M_{i+1})$ which is impossible.~\hfill~$\Box$
\end{proof}

The lemma above easily leads us to characterize monotone strategies as those ones for which there are no escape-doors.

\begin{corollary}\label{cor:monotone}
The strategy $\sigma$ is monotone if, and only if, for each vertex $v_i=(M_i,C_i)$ in $T(\sigma)$, and for each child $(M_{i+1},C_{i+1})$ of
$v_i$, $\E(v_i,M_{i+1})=\emptyset$.
\end{corollary}

Assume now that $\sigma$ is a non-monotone winning strategy. Armed with the above notions and results, we shall show how $\sigma$ can be
transformed into a monotone winning strategy, by ``removing'' the various escape-doors.

\begin{figure}[t]
  \centering
  \includegraphics[width=0.99\textwidth]{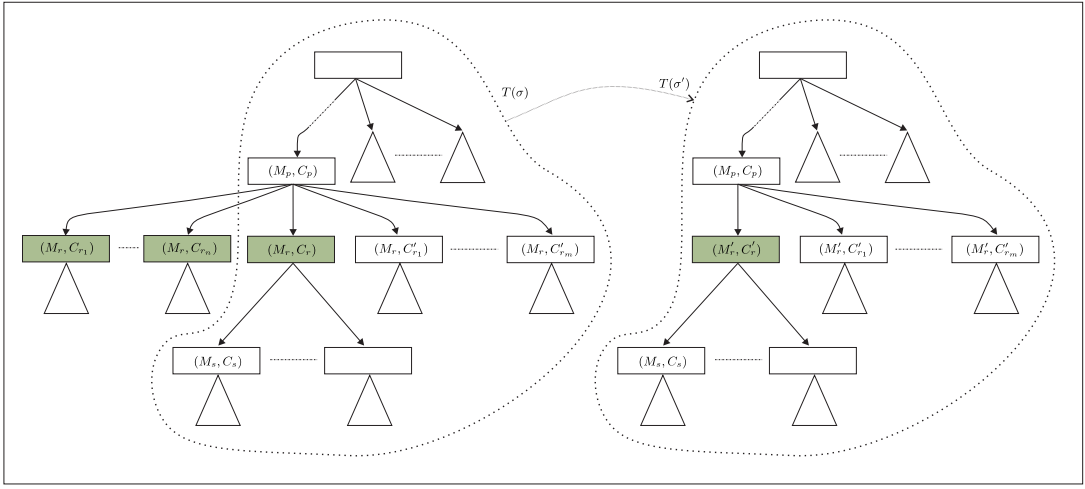}
  \caption{Exemplification in the proof of Lemma~\ref{lem:construction}.}\label{fig:nonmonotone2}
\end{figure}

Let $p=(M_p,C_p)$ be a configuration reached in $T(\sigma)$ from $(\emptyset,\nodes(\HG_1))$ by a (possibly empty) succession of moves $\pi$.
Assume that $M_r$ is the move of the Captain in $p$ and that this move is monotone, i.e., for each \ecomponent{p,M_r} $C$, $C\subseteq C_p$
(note that any move in the initial configuration is monotone). Let $r=(M_r,C_r)$ be a child of $p$ in $T(\sigma)$, and let $s=(M_s,C_s)$ be a
child of $r$ such that $C_s\not\subseteq C_r$, i.e., such that $\E(r,M_s)\neq \emptyset$ (by Corollary~\ref{cor:monotone}). This witnesses that
$M_s$ is a non-monotone move---see Figure~\ref{fig:nonmonotone2}.

Let $M_r'=M_r\setminus \E(r,M_s)\subset M_r$, and consider the function $\sigma'$ built as follows:
$$
\sigma'(M,C)= \left\{
\begin{array}{ll}
M_r' & \mbox{ if }(M,C)=(M_p,C_p)\\
\sigma(M,C) & \mbox{ otherwise.}\\
\end{array}
\right.
$$

Intuitively, we are removing from $r$ the source of non-monotonicity that was suddenly evidenced while moving to $s$, i.e., the fact that
$\E(r,M_s)\neq \emptyset$.
Next, we show that this modification does not affect the final outcome of the game.

\begin{lemma}\label{lem:construction}
$\sigma'$ is a winning strategy.
\end{lemma}
\begin{proof} By definition, $\sigma'$ leaves unchanged the configurations of $\sigma$ encoded in those subtrees
of $T(\sigma)$ that are not rooted below $p=(M_p,C_p)$---see Figure~\ref{fig:nonmonotone2}. Therefore, throughout this proof we have only to
take care of what happens in the subtree of $T(\sigma)$ rooted at $(M_p,C_p)$. Indeed, the first crucial difference between $\sigma'$ and
$\sigma$ occurs when the Captain plays $\sigma'(M_p,C_p)=M_r'$. Beforehand, we note that this is a ``valid'' move, since $M_r'\subset
M_r\subseteq \F(C_p)$ (cf. Fact~\ref{lem:main}).

Let $C_r'$ be the \component{M_r'} with $C_r\cup \E(r,M_s)\subseteq C_r'$. Note that such a component exists since $\E(r,M_s)\subseteq \F(C_r)$
and $M_r'\subseteq M_r$. In addition, we claim that $C_r'$ is a \ecomponent{p,M_r'}. Indeed, $C_r$ is a \ecomponent{p,M_r}, for which therefore
$C_r\cap (C_p\cup \E(p,M_r))\neq \emptyset$ holds, by Lemma~\ref{lem:escapeSpace}.
Moreover, since $M_r'\subset M_r$,  $\E(p,M_r)\subseteq \E(p,M_r')$ holds. Hence, given that $C_r\subseteq C_r'$,
it follows that $C_r'\cap (C_p\cup \E(p,M_r'))\neq \emptyset$, and again by Lemma~\ref{lem:escapeSpace}, that $C_r'$ is a \ecomponent{p,M_r'}.
Thus, there is an edge from $p$ to $(M_r',C_r')$ in $T(\sigma')$.
In order to complete the picture, we need the following two properties.\\

\noindent \textit{Property $P_1$:} \emph{For each \iecomponent{p,M_r} $C$, either $C\subseteq C_r'$ or $C$ is a \iecomponent{p,M_r'}.}

\vspace{-2mm}
\begin{itemize}
\item[\ ] \textbf{Proof.} We distinguish two cases depending on whether $\F(C)\cap \E(r,M_s)$ is empty or not.
In the case where $\F(C)\cap \E(r,M_s)=\emptyset$, then $C$ is an \component{M_r'} given that $M_r'=M_r\setminus \E(r,M_s)$. Moreover, by
Lemma~\ref{lem:escapeSpace}, $C$ is an \component{M_r} such that $C\cap (C_p\cup \E(p,M_r))\neq \emptyset$. Hence, it trivially holds that
$C\cap(C_p\cup \E(p,M_r'))\neq \emptyset$, because $\E(p,M_r)\subseteq \E(p,M_r')$. Thus, by Lemma~\ref{lem:escapeSpace}, $C$ is a
\ecomponent{p,M_r'}.
Eventually, consider the case where $\F(C)\cap \E(r,M_s)\neq \emptyset$, and recall that $C_r'$ is the \component{M_r'} with $C_r\cup
\E(r,M_s)\subseteq C_r'$. Since $M_r'=M_r\setminus \E(r,M_s)$, we then have that $C\subseteq C_r'$. \hfill $\diamond$
\end{itemize}

\noindent \textit{Property $P_2$:} \emph{For each \iecomponent{p,M_r'} $C'\neq C_r'$, $C'$ is a \iecomponent{p,M_r}.}

\vspace{-2mm}
\begin{itemize}
\item[\ ] \textbf{Proof.}  Let $C'\neq C_r'$ be a \ecomponent{p,M_r'}, hence in particular an \component{M_r'}. Since $\E(r,M_s)\subseteq
    C_r'$ and $C_r'$ is also an \component{M_r'}, we have that $C'\cap \E(r,M_s)=\emptyset$.
Moreover, $C'\cap M_r'=\emptyset$, with $M_r'=M_r\setminus \E(r,M_s)$. Thus, $C'\cap M_r\neq \emptyset$ and, hence, $C'$ is also an
\component{M_r}. Then, in the light of Lemma~\ref{lem:escapeSpace}, to conclude the proof, it suffices to show that $C'\cap (C_p \cup
\E(p,M_r))\neq \emptyset$ holds. In the case where $C'\cap C_p\neq \emptyset$, we have concluded.
Therefore, consider the case where $C'\cap C_p=\emptyset$. In this case, as $C'$ is a \ecomponent{p,M_r'} and, hence, $C'\cap (C_p\cup
\E(p,M_r'))\neq \emptyset$ because of Lemma~\ref{lem:escapeSpace}, we have $C'\cap \E(p,M_r')\neq \emptyset$.
Recall now that $\E(p,M_r)=M_p\cap \F(C_p)\setminus M_r$ and $\E(p,M_r')=M_p\cap \F(C_p)\setminus M_r'$.
Thus, $\E(p,M_r')\subseteq (M_r\setminus M_r')\cup \E(p,M_r)$, and then $\E(p,M_r')\subseteq \E(r,M_s)\cup \E(p,M_r)$ holds, as
$M_r'=M_r\setminus \E(r,M_s)$ by definition of the strategy $\sigma'$. Given that $C' \cap \E(r,M_s)=\emptyset$, we immediately can
conclude that $C'\cap \E(p,M_r')\subseteq C'\cap \E(p,M_r)$. However, $\E(p,M_r)\subseteq \E(p,M_r')$, because $M_r'\subset M_r$, and
hence, $C' \cap \E(p,M_r')= C'\cap \E(p,M_r)$ actually holds. Given that $C'\cap \E(p,M_r')\neq \emptyset$, we have therefore that $C'\cap
\E(p,M_r)\neq \emptyset$. Thus, $C'\cap (C_p \cup \E(p,M_r))\neq \emptyset$. \hfill $\diamond$
\end{itemize}

Note that in the light of the two results above, the function $\sigma'$ encodes a winning strategy when attacking each \ecomponent{p,M_r'}
$C'\neq C_r'$, since these components remain completely unchanged when changing $\sigma$ with $\sigma'$. Thus, as illustrated in
Figure~\ref{fig:nonmonotone2}, all subtrees of $T(\sigma)$ rooted at the children of $(M_p,C_p)$ and attacking options outside $C_r'$ are
preserved in the game-tree $T(\sigma')$. Hence, we have only to take care of how $\sigma'$ attacks the remaining component $C_r'$.

By definition of $\sigma'$, $\sigma'(M_r',C_r')=M_s$, which is a valid position since $M_s\subseteq \F(C_r')$ because of the facts that
$\sigma$ is a strategy (and, hence, we can apply Fact~\ref{lem:main} to conclude that $M_s\subseteq \F(C_r)$) and $C_r'\supseteq C_r$, so that
$\F(C_r)\subseteq \F(C_r')$. Moreover, the following property holds, which guarantees that the novel move is actually monotone.\\

\noindent \textit{Property $P_3$:} \emph{ $\E(r',M_s)=\emptyset$.}

\vspace{-2mm}
\begin{itemize}
\item[\ ] \textbf{Proof.} Assume by contradiction that $\E(r',M_s)= \partial C'_r \setminus M_s \neq\emptyset$, and let $X\in \E(r',M_s)$.
    From  $\partial C'_r \subseteq M'_r = M_r\setminus \E(r,M_s)$, we get that $X\in M_r \setminus \F(C_r)$. However, this is impossible
    because $M_r \subseteq \F(C_r)$, by definition of the Robber and Captain game.\hfill $\diamond$
\end{itemize}

We next show that applying $M_s$ to $C_r'$ leads exactly to the same strategy obtained when attacking $C_r$ with the same move $M_s$. Let $r'$
be the position $(M_r',C_r')$.\\

\noindent \textit{Property $P_4$:} \emph{  For each \iecomponent{r,M_s} $C$, $C$ is an \iecomponent{r',M_s}.}

\vspace{-2mm}
\begin{itemize}
\item[\ ] \textbf{Proof.} Let $C$ be an \ecomponent{r,M_s} and, hence, an \component{M_s} such that $C\cap(C_r\cup \E(r,M_s))\neq
    \emptyset$, by Lemma~\ref{lem:escapeSpace}. By definition, $C_r\cup \E(r,M_s)\subseteq C_r'$. Hence, $C\cap C_r'\neq \emptyset$. Again
    by Lemma~\ref{lem:escapeSpace}, we conclude that $C$ is a \ecomponent{r',M_s}.~\hfill~$\diamond$
\end{itemize}

\noindent \textit{Property $P_5$:} \emph{   For each \iecomponent{r',M_s} $C$, $C$ is an \iecomponent{r,M_s}.}

\vspace{-2mm}
\begin{itemize}
\item[\ ] \textbf{Proof.} Let $C$ be an \ecomponent{r',M_s} and, hence, an \component{M_s} such that $C\cap(C_r'\cup \E(r',M_s))\neq
    \emptyset$, by Lemma~\ref{lem:escapeSpace}. Our goal is to show that $C$ is also such that $C\cap(C_r\cup \E(r,M_s))\neq \emptyset$, so
    that $C$ is also an \ecomponent{r,M_s} (again by Lemma~\ref{lem:escapeSpace}). By Property $P_3$, $\E(r',M_s)=\emptyset$ and, hence,
    $C\cap C_r'\neq \emptyset$.
Let $Y$ be any vertex in $C\cap C_r'$ and assume, for the sake of contradiction, that $C \cap(C_r\cup \E(r,M_s))= \emptyset$. Then, since
$C$ is an \component{M_s}, $M_s$ separates $Y$ from the vertices in $C_r\cup \E(r,M_s)$, i.e., each path connecting $Y$ with some vertex in
$C_r\cup \E(r,M_s)$ must include a vertex belonging to $M_s$. 5 Let $\bar M_s\subseteq M_s$ be the set of all such vertices blocking the
paths from $Y$ to $C_r\cup \E(r,M_s)$. Then, consider the set $\{Y\}\cup C_r\cup \E(r,M_s)$ which is contained in $C_r'$. Hence, $\{Y\}\cup
C_r\cup \E(r,M_s)$ is \connected{M_r'}, because $C_r'$ is an \component{M_r'}. Therefore, the separator $\bar M_s$ cannot be included in
$M_r'$. It follows that there is a node $p\in \bar M_s\setminus M_r'$. Recall, now, that $M_r'=M_r\setminus \E(r,M_s)$. Thus, $p\not\in
M_r\setminus \E(r,M_s)$ holds. In addition, since $\E(r,M_s)=M_r\cap\F(C_r)\setminus M_s$, we have that $M_s\cap \E(r,M_s)=\emptyset$. So,
given that $p\in \bar M_s\subseteq M_s$, we have that $p\not\in \E(r,M_s)$. It follows that $p\not\in (M_r\setminus \E(r,M_s))\cup
\E(r,M_s)=M_r$.
Now observe that $M_s\subseteq \F(C_r)$ (because of Fact~\ref{lem:main}), while $\F(C_r)\setminus C_r\subseteq M_r$. Then, since $p\not\in
M_r$ and $p\in M_s$, we conclude that $p$ is in $C_r$. Thus, we can assume w.l.o.g that $Y$ is in $\F(C_r)$. Since $Y\not\in C_r\cup
\E(r,M_s)$, because of the assumption that $C \cap(C_r\cup \E(r,M_s))= \emptyset$ and the fact that $Y\in C$, we conclude that $Y\in
\F(C_r)\setminus C_r$, i.e., $Y\in M_r$. And, actually, $Y\in M_r'=M_r\setminus \E(r,M_s)$, given that $Y\not\in \E(r,M_s)$. But, this is
impossible, since $Y\in C_r'$ and $C_r'$ is an \component{M_r'}.~\hfill~$\diamond$
\end{itemize}

It follows that after the move $M_r'=\sigma'(M_p,C_p)$, the set of options available to the Robber is precisely the set that the Robber would
have obtained with the move $M_r$. Then, because $\sigma'$ attacks these options precisely as $\sigma$, we conclude that $\sigma'$ is still a
winning strategy. \hfill $\Box$
\end{proof}

\begin{figure}[t]
  \centering
  \includegraphics[width=0.98\textwidth]{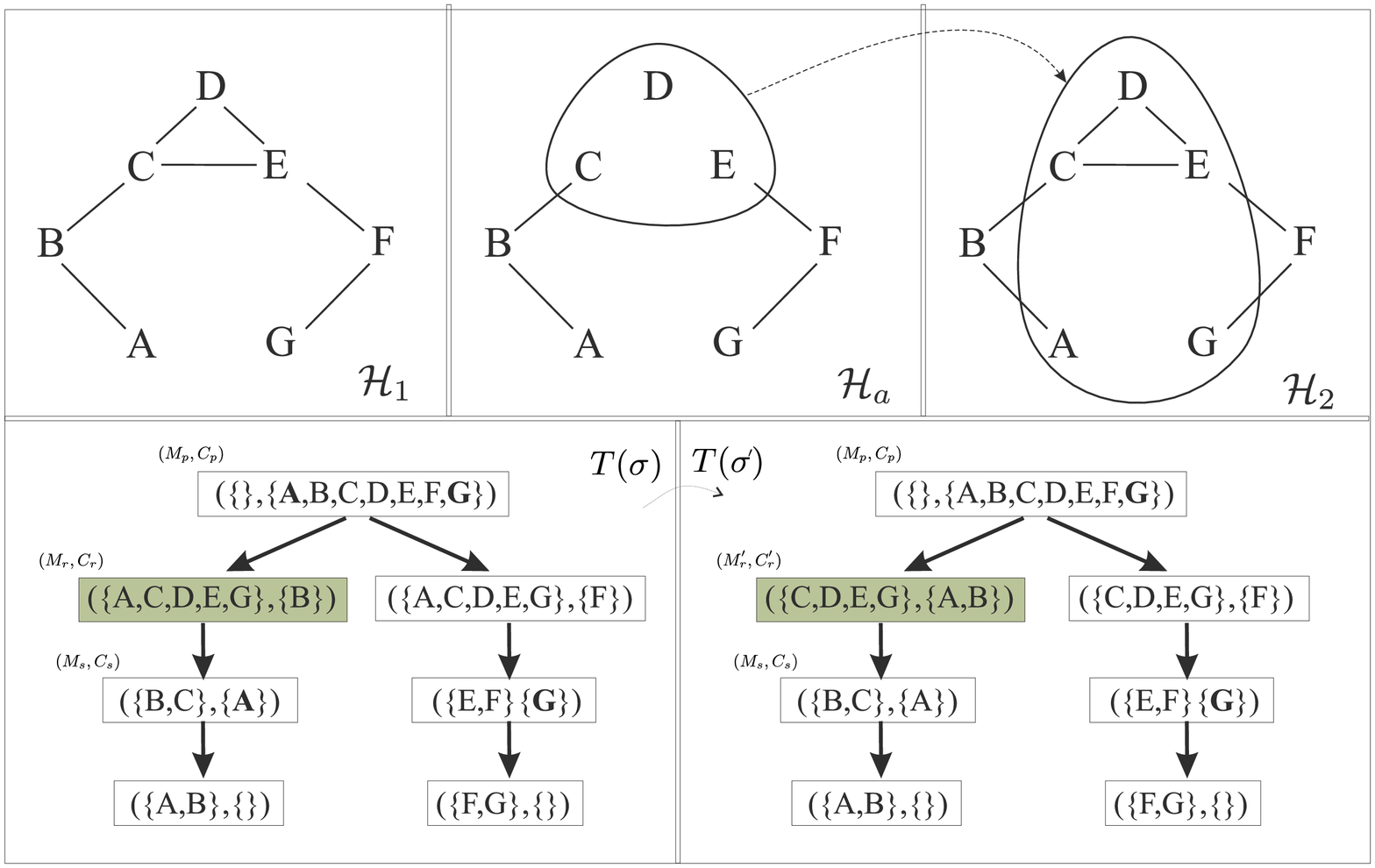}
  \caption{Illustration for Example~\ref{ex:construction}.}\label{fig:greedy-bis}
\end{figure}

\begin{example}\label{ex:construction}\em
For an example application of the above lemma, consider Figure~\ref{fig:greedy-bis}. The figure reports two hyperedges $\HG_1$ and $\HG_2$,
plus the game-tree for a winning strategy $\sigma$. In particular, note that $\sigma$ is non-monotone, because the Robber is allowed to return
on $A$ and $G$, after that these nodes have been previously occupied by the Captain with the move $M_r=\{A,C,D,E,G\}$. In fact, the figure also
reports the strategy $\sigma'$ that is obtained from $\sigma$, by turning the non-monotone move of the Captain (in the left branch of the tree)
into a monotone one according to the construction of Lemma~\ref{lem:construction}.

Note that the novel move of the Captain is $M_r'=\{C,D,E,G\}=M_r\setminus \{A\}$, with $\{A\}=\E(r',{B,C})$ being the escape door for the
Robber. In fact, this novel move does not affect the (winning) strategy of the Captain in the right branch.~\hfill~$\lhd$
\end{example}

Now that the transformation from $\sigma$ to $\sigma'$ has been clarified, we can state the main result of this section, which is based on the
fact that minimal winning strategies have to be monotone. For a strategy $\sigma$, let $||\sigma||$ be the \emph{size} of $\sigma$ measured as
total number of cops used over all the vertices of $T(\sigma)$, i.e., $||\sigma||=\sum_{(M,C)\in T(\sigma)}|M|$. Let $\sigma_1$ and $\sigma_2$
be two winning strategies. We write $\sigma_1\prec \sigma_2$ iff $||\sigma_1||<||\sigma_2||$. We say that a winning strategy $\sigma$ is
minimal, if there is no winning strategy $\bar \sigma$ such that $\bar \sigma\prec \sigma$. Note that the existence of a winning strategy
always entails the existence of a minimal winning one.

\begin{theorem}\label{thm:monotonization}
On the $\CR(\HG_1,\HG_2)$ game, the existence of a winning strategy implies the existences of a monotone winning strategy.
\end{theorem}
\begin{proof}
We claim that minimal winning strategies must be monotone. Indeed, let $\sigma$ be a non-monotone winning strategy and assume, for the sake of
contradiction, that $\sigma$ is minimal. Consider the transformation from $\sigma$ to $\sigma'$ discussed in the proof of
Lemma~\ref{lem:construction}, by recalling that $\sigma'$ is a winning strategy. Then, by definition of $\sigma'$ and the properties pointed
out in that proof, we have that $\sigma'\prec \sigma$, which is impossible. \hfill $\Box$
\end{proof}

As a remark, the transformation in the proof of Lemma~\ref{lem:construction} entails the existence of a constructive method to build a monotone strategy from a non-monotone one.

\subsection{Tree Projections and the R\&C Game}\label{sec:TreeProjections}

In this section, we prove that the Robber and Captain game precisely characterizes the tree projection problem, in the sense that a winning
strategy for $\CR(\HG_1,\HG_2)$ exists if and only if $(\HG_1,\HG_2)$ has a tree projection. Hence, any decomposition technique that can be
restated in terms of tree projections is in turn characterized by \CR\ games. In particular, if we consider pairs of the form
$(\HG_1,\HG_1^k)$, we get a game characterization for the notion of $k$-width generalized hypertree decomposition, for which such a
characterization was still missing in the literature.

For an exemplification of the results below, the reader may consider the game-tree illustrated in the bottom-right part of
Figure~\ref{fig:strategy} and the tree projection in Figure~\ref{fig:minimalTP}.

\begin{theorem}\label{thm:fromGameToTree}
If there is a winning strategy in $\CR(\HG_1,\HG_2)$, then $(\HG_1,\HG_2)$ has a tree projection.
\end{theorem}
\begin{proof}
From Theorem~\ref{thm:monotonization}, if there is a winning strategy in $\CR(\HG_1,\HG_2)$, there exists a monotone winning strategy, say
$\sigma$, for this game. Based on $\sigma$ we build a hypergraph $\HG_a(\sigma)$ where, for each vertex $(M,C)$ in $T(\sigma)$ with $M\neq
\emptyset$, $\edges(\HG_a(\sigma))$ contains the hyperedge $M$; and, no further hyperedge is in $\edges(\HG_a(\sigma))$. Note that, by
construction, $\HG_a(\sigma)\leq \HG_2$, since each position $M$ is such that $M\subseteq h_2$ for some hyperedge $h_2\in \edges(\HG_2)$.
Let $h_1$ be a hyperedge in $\edges(\HG_1)$. Since $\sigma$ is a winning strategy, we trivially conclude that the Captain has necessarily
covered in a complete form $h_1$ in some position. Thus, $\HG_1\leq \HG_a(\sigma)$. Eventually, in order to show that $\HG_a(\sigma)$ is a tree
projection, it remains to check that $\HG_a(\sigma)$ is acyclic.

To this end, we build a tree $\JT$ by exploiting the strategy $T(\sigma)$. $\JT$ contains all the hyperedges in $\edges(\HG_a)$ and, for each
pair of adjacent configurations $(M_s,C_s)$ and $(M_r,C_r)$ in $T(\sigma)$, the vertices $M_s$ and $M_r$ of $\JT$ are connected with an edge in
$\JT$. We claim that $\JT$ is a join tree for $\HG_a(\sigma)$. In the following, assume (for the sake of exposition) that $\JT$ is rooted at
the hyperedge encoding the first move of the Captain (e.g., $\{E,F,G\}$ in the game-tree depicted in Figure~\ref{fig:strategy}).

Note first that by construction of $\JT$ and since $\sigma$ is monotone, each vertex of $\JT$ has exactly one parent, but the root. Thus, $\JT$
is in fact an acyclic graph, and we have just to focus on showing that the connectedness condition is satisfied.

Let $h_1$ and $h_2$ be two distinct vertices in $\JT$ and let $X\in h_1\cap h_2$. Let $h$ be the vertex in the shortest path between $h_1$ and
$h_2$ that is the closest to the root of $\JT$. Assume, w.l.o.g., that $h_2\neq h$ and assume that $h_2$ is a child of $h'$. Because of
Fact~\ref{lem:main}, $h_2\subseteq \F(C_{h'})$, where $C_{h'}$ is the \component{h'} such that $(h,C_{h'})$ is in $T(\sigma)$ and where
$\sigma(h,C_{h'})=h_2$. Thus, $X\in \F(C_{h'})$.
Assume, for the sake of contradiction, that $X$ does not occur in a hyperedge in the path between $h_1$ and $h_2$, i.e., that the connectedness
condition is violated. In particular, w.l.o.g., we may just focus on the case where $X\not\in h'$. Then, since $h'\supseteq \F(C_{h'})\setminus
C_{h'}$, $X\in \F(C_{h'})$ and $X\not\in h'$ immediately entail that $X$ is in $C_{h'}$. Then, because of the monotonicity of $\sigma$, $X$
also occurs in an \component{h} $C_h$ such that $(h,C_h)\in T(\sigma)$. Now, observe that the scenario $h_1=h$ is impossible. Indeed, $h_1$
contains $X$ that would be also contained in an \component{h_1}, which is impossible by the monotonicity of the strategy. Thus, assume that
$h_1$ is the child of an edge $\bar h'$. By using the same line of reasoning as above, we conclude that $X$ occurs in an \component{h} $\bar
C_h$ such that $(h,\bar C_h)\in T(\sigma)$. Thus, $C_h=\bar C_h$. By definition of $\JT$, this means that $h_2$ and $h_1$ occur in a subtree
rooted at some child of $h$, which is impossible since $h$ is in the shortest path between $h_1$ and $h_2$. \hfill $\Box$
\end{proof}

We now complete the picture by showing the converse result.

\begin{theorem}\label{thm:fromTreetoGame}
If $(\HG_1,\HG_2)$ has a tree projection, then there is a winning strategy in $\CR(\HG_1,\HG_2)$.
\end{theorem}
\begin{proof}
Assume that $(\HG_1,\HG_2)$ has a tree projection. From Fact~\ref{thm:reducedMinimal}, it has a minimal tree projection, say $\HG_a$. Let $\JT$
be a join tree for $\HG_a$ in normal form (cf. Theorem~\ref{thm:mainMinimal}), and let $h$ be any hyperedge in $\edges(\HG_a)$.

Based on $\JT$, we build a strategy $\sigma$ as follows. Let $h_0=\emptyset$ and $C_0=\nodes(\HG_1)$. The first move of the Captain is $h$.
Recall from Section~\ref{MTP} that $\treecomp(h_s)$ is the unique \component{h_r} with $\nodes(\JT[h]_{h_s})=\treecomp(h_s)\cup (h_s\cap h_r)$,
where $h_s$ is a child of $h_r$ in $\JT[h]$ (with $\treecomp(h)$ be defined as $\nodes(\HG_a)$).
Given the current position $(h_p,C_p)$ and the current move $h_r$, assume that the following inclusion relationship holds: {for each
\ecomponent{(h_{p},C_{p}),h_r} $C_r$, $C_r\subseteq \treecomp(h_r)$}. Then, $\sigma(h_r,C_r)$ is defined as the hyperedge $h_{s}$ that is the
child of $h_r$ in $\JT[h]$ and that is such that $C_r=\treecomp(h_{s})$. It follows that the strategy $\sigma$ is well-defined under this
assumption, because such a hyperedge exists by Theorem~\ref{thm:mainMinimal}. Now, note that we can set $C_p=\treecomp(h_r)$ (for the first
move, just recall that $\treecomp(h)=\nodes(\HG_1)$).

We now show that the above inclusion relationship actually holds, that is, for each vertex $h_r$ of $\JT[h]$ and for each child $h_s$ of $h_r$,
we have $\treecomp(h_s)\subseteq \treecomp(h_r)$. To see this is true, recall again by Theorem~\ref{thm:mainMinimal} that
$\nodes(\JT[h]_{h_r})=\treecomp(h_r)\cup(h_r\cap h_p)$ and $\nodes(\JT[h]_{h_s})=\treecomp(h_s)\cup(h_s\cap h_r)$. Assume, for the sake of
contradiction, that $\treecomp(h_s)\not\subseteq \treecomp(h_r)$ and let $X\in\treecomp(h_s)\setminus \treecomp(h_r)$. Since,
$\nodes(\JT[h]_{h_s})\subseteq \nodes(\JT[h]_{h_r})$, it follows that $X\in h_r\cap h_p$. This is impossible, since $\treecomp(h_s)$ is a
\component{h_r}, with $X\in \treecomp(h_s)$.

Finally, to complete the proof just notice that the above also entails that $\sigma$ is a monotone strategy, eventually covering all the nodes
in $\HG_1$, hence it is a winning strategy.~\hfill~$\Box$
\end{proof}

\section{Applications and Conclusion}\label{sec:applications}

In this paper, we have analyzed structural decomposition methods to identify nearly-acyclic hypergraphs by focusing on the general concept of
tree projections.

We defined and studied a natural notion of minimality for tree-projections of pairs of hypergraphs. It turns out that minimal tree-projections
always exist (whenever some tree-projection exists), and that they enjoy some useful properties, such as the existence of join trees in a
suitable normal form that is crucial for algorithmic applications. In particular, such join trees have polynomial size with respect to the
given pair of hypergraphs.
As an immediate consequence of these properties, we get that deciding whether a tree projection of a pair of hypergraphs $(\HG_1,\HG_2)$ exists
is an $\NP$-problem. Note that this result is expected but not trivial, because in general a tree projection may employ any subset of every
hyperedge of $\HG_2$. In fact, the results proved in detail in the present paper have been (explicitly) used by \cite{GMS07} in the membership
part of the proof that deciding the existence of a tree projection is an $\NP$-complete problem, which closed the long-standing open question
about its computational complexity~\cite{goodman83synatctic,GS84,SS93,LS99}.

Moreover, we provided a natural game-theoretic characterization of tree projections in terms of the Captain and Robber game, which was missing
and asked for even in the special case of generalized hypertree decompositions. In this game, monotone strategies have the same power as
non-monotone strategies. Even this result is not just of theoretical interest. Indeed, by exploiting the power of non-monotonicity for some
easy-to-compute strategies in the game, called {\em greedy strategies}, larger islands of tractability for the homomorphism problem (hence, for
the constraint satisfaction problem and for the problem of evaluating conjunctive queries, and so on) have been identified in~\cite{GS12tr}. In
particular, for the special case of generalized hypertree decompositions, these strategies lead to the definition of a new tractable
approximation, called {\em greedy hypertree decomposition}, which is strictly more powerful than the (standard) notion of hypertree
decomposition.


\begin{thebibliography}{10}

\vspace{0.2cm}

\newcounter{c}
\newcommand{\NUM}{ }


\bibitem{adler04}\NUM I. Adler.
\newblock Marshals, monotone marshals, and hypertree-width.
\newblock {\em Journal of Graph Theory}, 47(4), pp. 275--296, 2004.

\bibitem{adler-thesis}\NUM I. Adler.
\newblock Width Functions for Hypertree Decompositions.
\newblock {\em PhD Thesis}, University of Freiburg, 2006.

\bibitem{adler08}\NUM I. Adler.
\newblock Tree-Related Widths of Graphs and Hypergraphs.
\newblock {\em SIAM Journal Discrete Mathematics}, 22(1), pp. 102--123, 2008.

\bibitem{AGG07}\NUM I. Adler, G. Gottlob, and M. Grohe.
\newblock Hypertree-Width and Related Hypergraph Invariants.
\newblock{ \em European Journal of Combinatorics}, 28, pp. 2167--2181, 2007.

\bibitem{ABD07}\NUM A. Atserias, A. Bulatov, and V. Dalmau.
\newblock On the Power of k-Consistency,
\newblock In {\em Proc. of ICALP'07}, pp. 279--290, 2007.


\bibitem{bern-good-81}\NUM P.A.~Bernstein and N.~Goodman.
\newblock The power of natural semijoins.
\newblock {\em SIAM Journal on Computing}, 10(4), pp. 751--771, 1981.

\bibitem{bodl-96}\NUM H.L.~Bodlaender and F.V.~Fomin.
\newblock A Linear-Time Algorithm for Finding Tree-Decompositions of Small Treewidth.
\newblock {SIAM Journal on Computing}, 25(6), pp. 1305-1317, 1996.

%
%
%

\bibitem{CJG08}\NUM D. A. Cohen, P. Jeavons, and M. Gyssens.
\newblock A unified theory of structural tractability for constraint satisfaction problems.
\newblock  {\em Journal of Computer and System Sciences}, 74(5), pp. 721-743, 2008.

%

\bibitem{DP89}\NUM R. Dechter and J. Pearl.
\newblock Tree clustering for constraint networks.
\newblock {\em Artificial Intelligence}, pp. 353--366, 1989.

\bibitem{fagi-83}\NUM R. Fagin.
\newblock Degrees of acyclicity for hypergraphs and relational database schemes.
\newblock {\em Journal of the ACM}, 30(3):514--550, 1983.


\bibitem{FN06}\NUM P.~Fraigniaud and N.~Nisse.
\newblock Connected Treewidth and Connected Graph Searching.
\newblock In {\em Proc. of LATIN'06}, pp. 479--490, 2006.

\bibitem{Fre90}\NUM E.C. Freuder.
\newblock Complexity of K-tree structured constraint satisfaction problems.
\newblock In {\em Proc. of the 8th National Conference on Artificial Intelligence}, pp. 4--9, 1990.

\bibitem{gott-etal-00}\NUM G.~Gottlob, N.~Leone, and F.~Scarcello.
\newblock A Comparison of Structural CSP Decomposition Methods.
\newblock \emph{Artificial Intelligence}, 124(2), 243--282, 2000.


\bibitem{gott-etal-99}\NUM G.~Gottlob, N.~Leone, and F.~Scarcello.
\newblock Hypertree decompositions and tractable queries.
\newblock {\em Journal of Computer and System Sciences}, 64(3), pp. 579--627, 2002.

\bibitem{gott-etal-03}\NUM G.~Gottlob, N.~Leone, and F.~Scarcello.
\newblock Robbers, marshals, and guards: game theoretic and logical characterizations of hypertree width.
\newblock {\em Journal of Computer and System Sciences}, 66(4), pp. 775--808, 2003.

\bibitem{GMS07}\NUM G. Gottlob, Z. Mikl\'os, and T. Schwentick.
\newblock Generalized hypertree decompositions: $\NP$-hardness and tractable variants.
\newblock {\em Journal of the ACM}, 56(6), 2009.


\bibitem{goodman83synatctic}\NUM N.~Goodman and O.~Shmueli.
\newblock Syntactic characterization of tree database schemas.
\newblock {\em Journal of the ACM}, 30(4):767--786, 1983.

\bibitem{GS84}\NUM N. Goodman and O. Shmueli.
\newblock The tree projection theorem and relational query processing.
\newblock {\em Journal of Computer and System Sciences}, 29(3), pp. 767--786, 1984.

\bibitem{GS08}\NUM G. Greco and F. Scarcello.
\newblock Tree Projections: Hypergraph Games and Minimality.
\newblock In {\em Proc. of ICALP'08}, pp. 736--747, 2008.


\bibitem{GS10}\NUM G. Greco and F. Scarcello.
\newblock Structural Tractability of Enumerating CSP Solutions.
\newblock In {\em Proc. of CP'10}, pp. 236--251, 2010.


\bibitem{GS12tr} G. Greco and F. Scarcello.
\newblock Tree Projections and Structural Decomposition Methods: The Power of Local Consistency and
Larger Islands of Tractability.
\newblock CoRR Technical report available at http://arxiv.org/abs/1205.3321, 2012. Full version of the paper
``The Power of Tree Projections: Local Consistency, Greedy Algorithms, and Larger Islands of Tractability'',
\newblock In {\em Proc. of PODS'10}, pp. 327--338, 2010.

\bibitem{G07}\NUM M.~Grohe.
\newblock The complexity of homomorphism and constraint satisfaction problems seen from the other side.
\newblock {\em Journal of the ACM}, 54(1), 2007.

\bibitem{GM06}\NUM M.~Grohe and D.~Marx.
\newblock Constraint solving via fractional edge covers.
\newblock In {\em Proc. of SODA '06}, pp. 289--298, 2006.

%
%
%

\bibitem{LS99}\NUM A.~Lustig and O.~Shmueli.
\newblock Acyclic hypergraph projections.
\newblock {\em Journal of Algorithms}, 30(2):400--422, 1999.


\bibitem{M10}\NUM D. Marx. Tractable Hypergraph Properties for Constraint Satisfaction and Conjunctive Queries. In \emph{Proc. of
STOC'10}, pp. 735--744, 2010.

\bibitem{RS84}\NUM N. Robertson and P.D. Seymour.
\newblock Graph minors III: Planar tree-width.
\newblock {\em Journal of Combinatorial Theory, Series B}, 36, pp. 49--64, 1984.


\bibitem{SS93}\NUM Y. Sagiv and O. Shmueli.
\newblock Solving Queries by Tree Projections.
\newblock {\em ACM Transaction on Database Systems}, 18(3), pp. 487--511, 1993.

%

\bibitem{ST93}\NUM P.D. Seymour and R. Thomas.
\newblock Graph searching and a min-max theorem for tree-width.
\newblock {\em Journal of Combinatorial Theory, Series B}, 58, pp. 22--33, 1993.


\bibitem{SH07}\NUM S.~Subbarayan and H.~Reif Andersen.
\newblock Backtracking Procedures for Hypertree, HyperSpread and Connected Hypertree Decomposition of CSPs.
\newblock In {\em Proc. of IJCAI'07}, pp. 180--185, 2007.

%
%

\end{thebibliography}
\end{document}